\documentclass[12pt,onecolumn,journal]{IEEEtran}

\usepackage{latexsym}
\usepackage[margin=2.54cm, lmargin=2.54cm]{geometry}
\usepackage{amsfonts}
\usepackage{amsbsy}
\usepackage{amsmath,amssymb}
\usepackage{amsthm}
\usepackage{times}
\usepackage{graphicx}

\usepackage{enumerate}
\usepackage[usenames]{color}
\usepackage[dvips]{pstcol}
\usepackage{epstopdf}
\usepackage{cite}
\usepackage{amsmath}
\usepackage{amssymb}
\usepackage{amsfonts}
\usepackage{graphicx}
\usepackage{epsfig}
\usepackage{ulem}
\usepackage{psfrag}
\usepackage{xcolor}
\usepackage{amsfonts, bm}
\usepackage{epstopdf}
\usepackage{cite}
\usepackage{color}
\usepackage{xcolor}
\usepackage{verbatim}
\usepackage{subfig}
\usepackage{algorithm}
\newtheorem{theorem}{Theorem}

\usepackage{verbatim}
\input epsf

\begin{document}
\title{ MIMO-Aided Nonlinear Hybrid Transceiver Design for Multiuser mmWave Systems Relying on Tomlinson-Harashima Precoding }
\author{Kaidi Xu, Yunlong Cai, Minjian Zhao, Yong Niu, and Lajos Hanzo
\thanks{
K. Xu, Y. Cai, and M. Zhao    are with the College of Information Science and Electronic Engineering, Zhejiang University, Hangzhou 310027, China (e-mail: xukaidi13@126.com; ylcai@zju.edu.cn; mjzhao@zju.edu.cn).

Y. Niu is with the State Key Laboratory of Rail Traffic Control and Safety, Beijing Jiaotong University, Beijing 100044, China (e-mail: niuy11@163.com).

L. Hanzo is with the Department of ECS, University of Southampton, U.K. (e-mail:
lh@ecs.soton.ac.uk).
}
}

\maketitle
\vspace{-4.8em}

\begin{abstract}
  Hybrid analog-digital (A/D) transceivers designed for millimeter
  wave (mmWave) systems have received substantial research attention,
  as a benefit of their lower cost and modest energy consumption
  compared to their fully-digital counterparts.  We further improve
  their performance by conceiving a Tomlinson-Harashima precoding
  (THP) based nonlinear joint design for the downlink of multiuser
  multiple-input multiple-output (MIMO) mmWave systems. Our
  optimization criterion is that of minimizing the mean square error
  (MSE) of the system under channel uncertainties subject both to
  realistic transmit power constraint and to the unit modulus
  constraint imposed on the elements of the analog beamforming (BF)
  matrices governing the BF operation in the radio frequency domain.  We
  transform this optimization problem into a more tractable form and
  develop an efficient block coordinate descent (BCD) based algorithm
  for solving it.  Then, a novel two-timescale nonlinear joint hybrid
  transceiver design algorithm is developed, which can be viewed as an
  extension of the BCD-based joint design algorithm for reducing both
  the channel state information (CSI) signalling overhead and the
  effects of outdated CSI.  Moreover, we determine the near-optimal
  cancellation order for the THP structure based on the lower bound of
  the MSE.  The proposed algorithms can be guaranteed to converge to a
  Karush-Kuhn-Tucker (KKT) solution of the original problem.  The
  simulation results demonstrate that our proposed nonlinear joint
  hybrid transceiver design algorithms significantly outperform the
  existing linear hybrid transceiver algorithms and approach the
  performance of the fully-digital transceiver, despite its lower
  cost and power dissipation.
\end{abstract}
\vspace{-1.5em}
\begin{IEEEkeywords}
Nonlinear precoding, hybrid A/D beamforming, hardware-efficient, mmWave, two-timescale.
\end{IEEEkeywords}

\section{Introduction}
The global spectrum shortage has stimulated considerable interest in
the development of millimeter wave (mmWave) communications for the
next generation wireless networks~\cite{pi2011introduction,
  rappaport2013millimeter, heath2016overview, hemadeh2018millimeter,
  he2020propagation, he2018geometrical}.  At a carrier frequency of
30~GHz or 1cm wavelength, numerous antenna elements can be packed into
a compact space. This facilitates large-scale spatial multiplexing and
high-gain directional beamforming (BF) and thereby significantly
increases the system capacity.  However, for large-scale
multiple-input multiple-output (MIMO) mmWave systems the conventional
fully-digital (FD) BF architecture requires numerous radio frequency
(RF) chains which results in extremely high fabrication cost and high
power consumption.  In order to circumvent these drawbacks, hybrid
analog-digital (A/D) BF architectures have been proposed, which
require less RF chains than the FD BF architecture, when nusing the
same number of antennas \cite{liu2017millimeter, li2019explore,
  sohrabi2016hybrid, yu2016alternating, ni2017near, chen2015iterative,
  mendez2015dictionary, ayach2014spatially, alkhateeb2015limited,
  zhai2017joint,he2017codebook,lin2017subarray, shi2018spectral,
  cai2019robust, liu2019stochastic, chen2019randomized,
  shi2017penalty}.


In~\cite{liu2017millimeter}, the authors analyzed the beam-alignment
performance of both exhaustive and hierarchical search techniques,
with the time-domain training overhead taken into account.  An
optimized two-stage search algorithm was proposed in
\cite{li2019explore} for transmitter and receiver beam alignment.  In
\cite{sohrabi2016hybrid}, the authors established that a hybrid A/D BF
structure with twice as many RF chains as data streams is capable of
realizing any FD BF structure exactly.  A series of
matrix-decomposition based hybrid BF design algorithms have been
proposed in \cite{ni2017near, chen2015iterative,
  mendez2015dictionary}.  By exploiting the sparse nature of the
channel matrix, the authors of~\cite{ayach2014spatially} formulated
the hybrid BF design problem as a sparse matrix reconstruction problem
and solved it using the modified orthogonal matching pursuit (OMP)
algorithm.  In~\cite{alkhateeb2015limited}, the authors investigated a
hybrid transceiver design using realistic limited feedback in their
multi-user mmWave systems.  As a further advance, the authors
of~\cite{yu2016alternating} have developed an alternating minimization
algorithm for their hybrid BF design with the aid of manifold
optimization (MO).  The authors of~\cite{zhai2017joint} considered the
uplink of large-scale multiuser MIMO mmWave systems, where the
implementation cost of their joint hybrid BF algorithm was reduced
with the aid of antenna selection.  In order to mitigate the
hardware-induced performance erosion, a number of codebook-based
hybrid BF algorithms were conceived
in~\cite{he2017codebook,lin2017subarray}.  In
\cite{shi2018spectral,cai2019robust} the unit-modulus constraint and
power constraints imposed upon the A/D hybrid BF were mitigated by the
penalty dual decomposition (PDD)~\cite{shi2017penalty} based hybrid BF
design algorithm, which can be guaranteed to achieve the
Karush-Kuhn-Tucker (KKT) solution.  In particular, the authors of
\cite{shi2018spectral} directly optimized the spectral efficiency of
the mmWave downlink in a multiuser multistream MIMO system. Then Cai
\textit{et al.}~\cite{cai2019robust} extended the solution advocated
in~\cite{shi2018spectral} to a mmWave full-duplex MIMO relay-aided
system.  As another development, both the channel state information
(CSI) feedback overhead and the implementation complexity were reduced
as part of a series of two-timescale based studies for the design of
A/D hybrid BF~\cite{liu2019stochastic, chen2019randomized,
  mai2018two}. Explicitly, the long-timescale analog BF matrices were
optimized based on the channel statistics, while the short-timescale
digital precoding matrices were updated according to the
near-instantaneous CSI.



In parallel to the low-complexity linear transceiver structures, more
sophisticated nonlinear transceivers, such as the Tomlinson-Harashima
precoding (THP) have also evolved from the seminal contributions
of~\cite{tomlinson1971new,harashima1972matched}, leading to powerful
spatial-domain MIMO solutions~\cite{fischer2002space}.  The THP-based
nonlinear transceiver algorithms have a remarkable performance gain
over their linear counterparts and thus have found numerous
applications \cite{geng2015robust,zarei2018robust,tseng2011joint}.
However, determining the optimal cancellation order under the THP
structure, which achieves the optimal performance gain is quite a
challenge~\cite{liu2007improved}.  In \cite{zu2014multi}, the authors
proposed a multi-branch (MB) THP scheme, where each branch contains
a THP with a predefined ordering strategy, and a selection criterion
is applied to choose the best branch to generate the final output.
Moreover, the THP-based robust nonlinear transceiver design has also
been further developed by taking the CSI errors into account in
relay-aided multiuser MIMO systems~\cite{zhang2014robust}. This
solution has also been extended to a full-duplex relay-aided wireless
power transfer system in~\cite{zhang2017nonlinear}.  Finally, the
authors of~\cite{gao2018optimization, masouros2012interference}
proposed techniques for reducing the power-loss imposed by the modulo and
feedback operations used in the THP.

However, to the best of our knowledge, the aforementioned A/D hybrid
transceiver design algorithms are all based on the linear precoding
structure, which suffers from the performance degradation caused by
the multiuser interference and by the reduced number of available RF
chains.  Against this background, we propose a THP-based joint A/D
hybrid transceiver design algorithm for the downlink of multiuser
mmWave MIMO systems for further improving the system
performance. Specifically, we jointly optimize the analog BF matrices
and the digital processing matrices, i.e., the digital precoding and
the receiver as well as feedback matrices of the THP
structure. Explicitly, we minimize the system's mean square error (MSE)
subject to both the transmit power constraint and the unit modulus
constraint imposed on each element of the analog RF BF matrices.  The
optimization problem formulated is quite challenging to tackle.  By
efficiently exploiting the particular structure of this problem, we first
transform it into a more tractable form.  Then we propose an
efficient block coordinate descent (BCD) based algorithm for solving the
converted problem.  Furthermore, we extend the proposed BCD-based
joint design algorithm to a novel two-timescale nonlinear joint hybrid
transceiver design algorithm in order to reduce both the CSI signalling
overhead and the effects of outdated CSI caused by its feedback delay.
The proposed algorithms can be guaranteed to obtain a KKT solution of
the original problem.

The main contributions of this work are summarized as follows:
\begin{enumerate}


\item There is a paucity of literature on optimizing the nonlinear A/D
  hybrid transceiver matrices by minimizing the MSE, because this
  problem is very challenging.  Hence we first transform this problem into a
  more tractable form and optimize the matrix variables in a BCD
  fashion, where the subproblems of each block can be solved in closed
  form.



\item

We develop a novel two-timescale nonlinear hybrid transceiver design
algorithm based on two-stage online successive convex approximation
(TOSCA). Although the proposed TOSCA-based two-timescale algorithm
suffers from a certain performance degradation compared to the
proposed BCD-based joint design algorithm in the presence of small
delays, both the CSI signalling overhead and the effects of outdated CSI
caused by high CSI-feedback delays can be substantially reduced. In
this scheme, the long-timescale analog BF matrices are optimized based
on the channel statistics, while the short-timescale digital processing
matrices are designed based on the low-dimensional effective CSI
matrices for each time slot.


\item We determine the near-optimal cancellation order for the
  proposed THP-based hybrid transceiver design based on the lower
  bound of the MSE.  Our simulation results demonstrate that the
  proposed BCD-based joint nonlinear hybrid transceiver design
  algorithm significantly outperforms the existing linear hybrid
  transceiver algorithms and approaches the performance of the
  fully-digital transceiver. Furthermore, compared to the proposed
  BCD-based joint design algorithm, the proposed two-timescale joint
  design algorithm provides better performance in the scenario of
  severe CSI delays, although it suffers from some performance
  degradation for small delays.
\end{enumerate}

The rest of this paper is structured as follows. Section
\ref{system_model} introduces the proposed THP-based mmWave multiuser
MIMO system and the optimization problems formulated.  In Section
\ref{HBF_design}, we first transform the problem into a more tractable
form and then propose a BCD-based joint design algorithm to solve it.
In Section \ref{TTS_design}, we propose the TOSCA-based two-timescale
joint nonlinear transceiver design algorithm.
In Section \ref{cancellation_order}, we derive the lower bound of the
MSE and determine the near-optimal cancellation order for the proposed
THP-based hybrid transceiver design. Our simulation results are presented
in Section \ref{simulation_results}. Finally, Section \ref{conclusion}
offers our conclusions.

\emph{Notations:} Scalars, vectors and matrices are respectively
denoted by lower case, boldface lower case and boldface upper case
letters.  $\mathbf{I}$ represents an identity matrix and $\mathbf{0}$
denotes an all-zero matrix.  For a matrix $\mathbf{A}$,
${{\bf{A}}^T}$, $\mathbf{A}^*$, ${{\bf{A}}^H}$ and $\|\mathbf{A}\|$
denote its transpose, conjugate, conjugate transpose and Frobenius
norm, respectively. For a square matrix $\bf{A}$, $\textrm{Tr}
(\bf{A})$ denotes its trace, ${\bf{A}} \succeq {\bf{0}}~({\bf{A}}
\preceq {\bf{0}})$ means that $\bf{A}$ is positive (negative)
semidefinite.
$[\mathbf{A}]_{a:b\,,c:d}$ represents a submatrix of $\mathbf{A}$.
For a vector $\mathbf{a}$, $\|\mathbf{a}\|$ represents its Euclidean norm.
$\mathbb{E}\{.\}$ denotes the statistical expectation.
$\text{Re}(.)$ ($\text{Im}(.)$) denotes
the real (imaginary) part of a variable.
The operator $\textrm{vec}( \cdot )$ stacks the elements of a matrix in one long column vector.
 $|  \cdot  |$ denotes  the absolute value of a complex scalar.
 The operator $\angle$ takes the phase angles of the elements in a matrix.
${\mathbb{C}^{m \times n}}\;({\mathbb{R}^{m \times n}})$ denotes the space of ${m \times n}$ complex (real) matrices.
The symbol $\otimes$ denotes the Kronecker product of two vectors/matrices.

\begin{figure*}	
	\centering
	\includegraphics[scale=0.7]{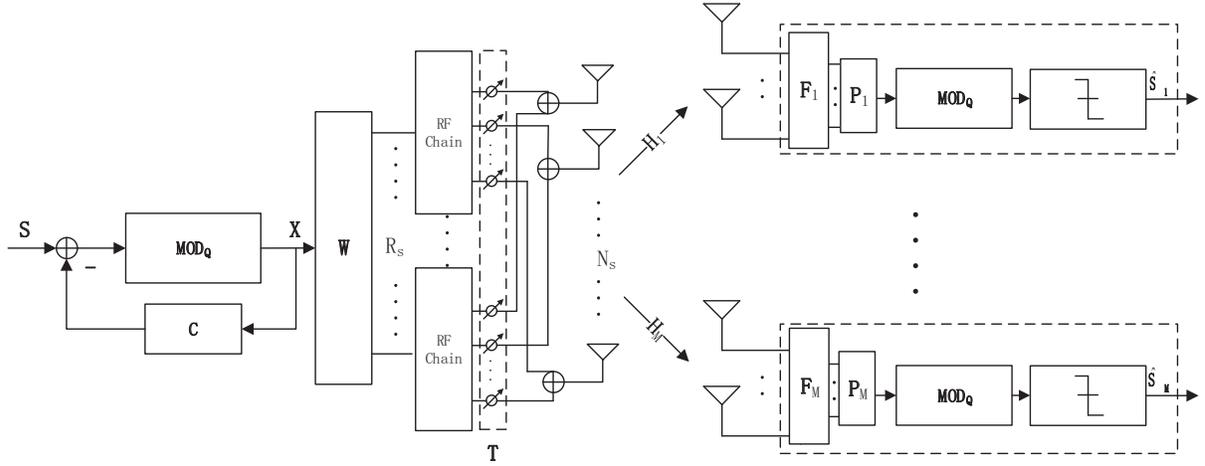}	
	\caption{System model}
	\label{SystemModel}
\end{figure*}
\section{System model and problem formulation}
\label{system_model}

In this section, we first introduce the system model of mmWave
multiuser MIMO systems, and then mathematically formulate the
optimization problem of interest.
\subsection{System model}

As illustrated in Fig. \ref{SystemModel}, we consider a mmWave
communication system comprising of one BS and $M$ users.  The BS
equipped with $N_s$ antennas and $R_s$ RF chains, where $N_s\geq R_s$,
transmits a signal vector
$\mathbf{s}=[\mathbf{s}_1^T,\mathbf{s}_2^T,\ldots,\mathbf{s}_M^T]^T\in\mathbb{C}^{D\times
  1}$ to the users, where $\mathbf{s}_m \in \mathbb{C}^{D_{m}\times
  1}$ denotes the signal vector for user $m$,
$m\in\mathcal{M}\triangleq \{1, 2, ..., M\}$, and $D = \sum_{m=1}^M
D_m$ denotes the total number of transmit data streams.  User $m$ is
equipped with $N_{d,m}$ antennas and $R_{d,m}$ RF chains. Besides, we
assume the necessary condition $R_s\geq D$ for sufficient degree of
freedom.  Each entry of the transmitted signal vector $\mathbf{s}$ is
a Q-ary quadrature amplitude modulation (QAM) signal.  Hence the real
and imaginary parts of each entry of the signal vector $\mathbf{s}$
are independent and identically distributed (i.i.d.) random variables
generated from a Q-ary QAM alphabet $\mathcal{A}$.  Specifically, we
let $\mathcal{A} =
\{\pm\sqrt{\frac{3}{2(Q-1)}},\pm3\sqrt{\frac{3}{2(Q-1)}},
\pm5\sqrt{\frac{3}{2(Q-1)}}, \ldots,
\pm(\sqrt{Q}-1)\sqrt{\frac{3}{2(Q-1)}}\}$,
$\text{E}\{\mathbf{s}\}=\mathbf{0}$ and
$\text{E}\{\mathbf{s}\mathbf{s}^H\}=\mathbf{I}$, where
$\text{Re}(s_k)\in\mathcal{A},\,\forall k$,
$\text{Im}(s_k)\in\mathcal{A},\,\forall k$ and $s_k$ is the $k$th
entry of the vector $\mathbf{s}$.

$\text{MOD}_Q(.)$ in Fig. \ref{SystemModel} is a modulo operator used
to constrain a value in $(-\sqrt{\tau},\sqrt{\tau}]$, where $\tau =
  \sqrt{\frac{3Q}{2(Q-1)}}$. This operator can be formulated as
\begin{equation}
	\text{MOD}_Q(x)=x-2\sqrt{\tau}\lfloor \frac{x+\sqrt{\tau}}{2\sqrt{\tau}} \rfloor = x+e,
	\label{mod}
\end{equation}
where $e$ is the residual error.

With the modulo operator in \eqref{mod}, we can generate the transmit
symbols $x_k$ successively as
\begin{equation}	\
	x_k = s_k-\sum_{n=1}^{k-1}[\mathbf{C}]_{k,n}x_n + e_k,
	\label{channalsymbol}
\end{equation}
where $\mathbf{C} \in \mathbb{C}^{D \times D }$ is a strictly lower
triangle matrix and $\mathbf{e} = [e_1,e_2,\ldots,e_M]^T$ is the
residual error vector generated by the modulo operator. Equation
\eqref{channalsymbol} can be rewritten in a matrix form as
\begin{equation}
	\mathbf{x}=\mathbf{U}^{-1}\mathbf{v},
\end{equation}
where $\mathbf{U} = \mathbf{I}+\mathbf{C}$ is a lower triangle matrix
with ones on the main diagonal and
$\mathbf{v}=\mathbf{s}+\mathbf{e}=[\mathbf{v}_1^T,\mathbf{v}_2^T,\ldots,\mathbf{v}_M^T]^T$
is the target signal vector\footnote{This is because the users can
  reconstruct $\mathbf{s}$ form $\mathbf{v}$ via
  $\mathbf{s}=\text{MOD}_Q(\mathbf{v})$.}.  Together with the
consideration in \cite{shenouda2008framework}, we have
$\text{E}\{\mathbf{x}\mathbf{x}^H\}=\mathbf{I}$ for a high order
$Q$-ary QAM constellation.

Before transmission, the processed signal $\mathbf{x}$ is passed
through a linear digital precoding matrix $\mathbf{W} \in
\mathbb{C}^{R_s \times D}$ followed by an analog BF matrix $\mathbf{T}
\in \mathbb{C}^{N_s \times R_s}$. The transmit power constraint at the
BS is given by
\begin{equation}
\text{E}\{\|\mathbf{T}\mathbf{W}\mathbf{x}\|^2\}=\|\mathbf{TW}\|^2\leq P_t,\label{power_constraint}
\end{equation}
where $P_t$ is the transmit power budget. The signal received at user $m$ is given by
\begin{equation}
	\mathbf{y}_m=\mathbf{H}_m\mathbf{TWx}+\mathbf{n}_m,
\end{equation}
where $\mathbf{H}_m\in \mathbb{C}^{N_{d,m}\times N_s}$ denotes the
MIMO channel matrix between the BS and user $m$, while $\mathbf{n}_m$
denotes the complex-valued circular Gaussian noise at user $m$ with
zero mean and correlation matrix
$\text{E}\{\mathbf{n}_m\mathbf{n}_m^H\}=\sigma_m^2 \mathbf{I}$.

At user $m$, a linear A/D hybrid receiver consisting of an analog BF
matrix $\mathbf{F}_m \in \mathbb{C}^{R_{d,m} \times N_{d,m}}$ and a
digital receiving matrix $\mathbf{P}_m \in \mathbb{C}^{D_m \times
  R_{d,m}}$ is employed for detecting symbols.  The output of the
hybrid receiver is expressed as
\begin{equation}
	\hat{\mathbf{v}}_m=\mathbf{P}_m\mathbf{F}_m\mathbf{H}_m\mathbf{TWx}+\mathbf{P}_m\mathbf{F}_m\mathbf{n}_m,
\end{equation}
while the final estimate of the signal vector for user $m$ is given by
\begin{equation}
	\hat{\mathbf{s}}_m=\text{MOD}_Q(\hat{\mathbf{v}}_m).
\end{equation}

In practice, channel estimation errors are inevitable.  According to
\cite{lee2012effect}, the channel estimation errors can be modelled as
\begin{equation}
	\mathbf{H}_m=\bar{\mathbf{H}}_m+\sigma_{e,m}\Delta\mathbf{H}_m \quad\forall m,
\end{equation}
where $\bar{\mathbf{H}}_m \in \mathbb{C}^{N_{d,m}\times N_s}$ denotes
the estimated channel matrix, $\Delta\mathbf{H}_m$ denotes the channel
estimation error matrix, and $\sigma_{e,m}$ denotes the estimation
error variance. Specifically, $\Delta\mathbf{H}_m$ is i.i.d. with
zero-mean and unit-variance circular complex Gaussian distribution.

Furthermore, the ordering scheme for the THP structure is considered
as a matrix $\mathbf{L}\in \mathbb{R}^{D\times D}$ whose elements are
zeros and ones.  The ordering matrix $\mathbf{L}$ follows the
constraints $\mathbf{L1}=\mathbf{1}$,
$\mathbf{1}^T\mathbf{L}=\mathbf{1}^T$, that is, in each row and column
only one entry is 1 and the others are 0s.  Hence the permutation
process can be expressed as $\mathbf{s}=\mathbf{L}\mathbf{\tilde{s}}$,
where $\mathbf{\tilde{s}}$ is the original transmit data vector and
$\mathbf{s}$ is the permutated data vector.  Then we have the desired
output signal vector of the linear receiver for user $m$ is
$\mathbf{\tilde{v}}_m = \mathbf{A}_m\mathbf{L}^T\mathbf{v}$, where
$\mathbf{A}_m=[\mathbf{0}_{D_{m}\times\sum_{i=1}^{m-1}D_i},\mathbf{I}_{D_m},\mathbf{0}_{D_m\times\sum_{i=m+1}^MD_i}]$
denotes a selection matrix extracting the entries of user $m$ in
vector $\mathbf{L}^T\mathbf{v}$.\footnote{This is because the desired
  signal vector for user $m$ is
  $\mathbf{A}_m\mathbf{L}^T\mathbf{s}=\mathbf{A}_m\mathbf{L}^T\text{MOD}_Q(\mathbf{v})=\text{MOD}_Q(\mathbf{A}_m\mathbf{L}^T\mathbf{v})$.}
Finally, the MSE at user $m$ can be expressed as
\begin{equation}
\begin{aligned}
	\text{MSE}(\{\mathbf{P}_m,\mathbf{F}_m\},\mathbf{T},\mathbf{W},\mathbf{U})=&\sum_{m=1}^M\text{E}\{\|\hat{\mathbf{v}}_m - \tilde{\mathbf{v}}_m\|^2\}\\
	=&\sum_{m=1}^M \text{tr}\bigg(\mathbf{P}_m\mathbf{F}_m\mathbf{\bar{H}}_m\mathbf{TWW}^H\mathbf{T}^H\mathbf{\bar{H}}_m^H\mathbf{F}_m^H\mathbf{P}_m^H\\
	&+\sigma_{e,m}^2\text{tr}(\mathbf{TWW}^H\mathbf{T}^H)\mathbf{P}_m\mathbf{F}_m\mathbf{F}_m^H\mathbf{P}_m^H +\sigma_m^2\mathbf{P}_m\mathbf{F}_m\mathbf{F}_m^H\mathbf{P}_m^H\\
	&-\mathbf{P}_m\mathbf{F}_m\bar{\mathbf{H}}_m\mathbf{TWU}^H\mathbf{L}\mathbf{A}_m^H-\mathbf{A}_m\mathbf{L}^T\mathbf{U}\mathbf{W}^H\mathbf{T}^H\mathbf{\bar{H}}_m^H\mathbf{F}_m^H\mathbf{P}_m^H\bigg)\\
	&+\text{tr}(\mathbf{U}\mathbf{U}^H),\label{Lobj}
\end{aligned}
\end{equation}
where the expectation here is taken over the random variables $\{\Delta\mathbf{H}_m,\mathbf{n}_m\}$.

\subsection{Problem formulation}
\subsubsection{Joint design problem}
With the expression of MSE shown in \eqref{Lobj} and the power
constraint shown in \eqref{power_constraint}, we are now able to
formulate the proposed THP-based hybrid transceiver design problem.
We aim to jointly design the digital precoding and feedback matrices
in the THP structure and the analog BF matrices to minimize the MSE,
hence this problem can be formulated as follows
\begin{subequations}
\begin{align}
	\min_{\{\mathbf{P}_m,\mathbf{F}_m\},\mathbf{T},\mathbf{W},\mathbf{U}}&\quad \text{MSE}(\{\mathbf{P}_m,\mathbf{F}_m\},\mathbf{T},\mathbf{W},\mathbf{U})\\
	s.t.\qquad &|[\mathbf{F}_m]_{i,j}|=1 \quad \forall m,i,j,\label{abs_constraint1}\\	
	&|[\mathbf{T}]_{i,j}|=1 \quad \forall i,j,\label{abs_constraint2} \\
    & \|\mathbf{TW}\|^2\leq P_t ,\label{power_constraint10}
\end{align}\label{original_problem}%
\end{subequations}
where the constant modulus constraints given by
\eqref{abs_constraint1} and \eqref{abs_constraint2} are due to the
fact that the analog beamformer is implemented using low-cost phase
shifters.

\subsubsection{Two-timescale joint design problem}
In practice, the analog BF matrices can also update over a longer
timescale than the digital processing matrices aiming at reducing the
feedback overhead needed for the exchange of CSI. Specifically, the
long-timescale variables, i.e., the analog BF matrices, are designed
based on the slowly varying channel statistics\footnote{The channel
  statistics refer to the distribution of channel fading
  realizations. We only need to obtain a single (potentially outdated)
  channel sample at each frame, based on which the analog BF matrices
  can be updated directly.} while the short-timescale variables, i.e.,
the digital processing matrices, are optimized based on the
instantaneous effective low-dimensional CSI matrices.

In particular, as illustrated in Fig. \ref{timeline}, the time axis is
divided into some super-frames within which the channel statistics
remains coherent.  Each super-frame consists of $T_f$ frames, each of
which is made up of $T_s$ time slots.  Within each time slot, the
instantaneous effective CSI remains unchanged.  During the
implementation of our proposed two-timescale algorithm, the
long-timescale variables are updated at the end of each frame based on
a channel sample, while the short-timescale variables are updated at
the beginning of each time slot based on the instantaneous effective
CSI.  Consequently, we formulate our two-timescale optimization
problem as
\begin{equation}
\begin{aligned}
	\min_{\{\mathbf{F}_m\},\mathbf{T},\boldsymbol{\Theta}}&\quad f(\{\mathbf{F}_m\},\mathbf{T},\boldsymbol{\Theta})\triangleq\mathbb{E}_{\mathbf{\bar{H}}_m}\{\text{MSE}(\{\mathbf{P}_m,\mathbf{F}_m\},\mathbf{T},\mathbf{W},\mathbf{U})\}\\
	\text{s.t.}&\quad \eqref{abs_constraint1}-\eqref{power_constraint10},
\end{aligned}\label{TTS_problem}
\end{equation}
where $\boldsymbol{\Theta}\triangleq\{\{\mathbf{P}_m\},\mathbf{W},\mathbf{U}\}$ denotes a collection of the short-timescale variables and the expectation here is taken over the channel samples $\{\mathbf{\bar{H}}_m\}$ within a super-frame.

\begin{figure}
\centering
\includegraphics[scale=0.7]{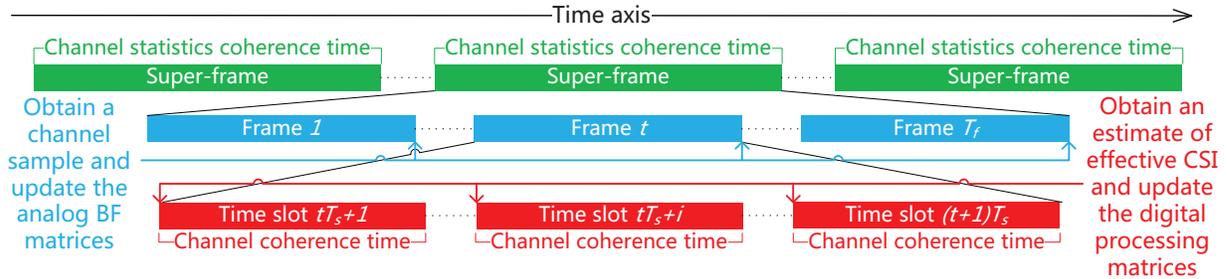}
\caption{Decomposition of the time axis into two timescales}
\label{timeline}
\end{figure}

\section{Proposed Hybrid Transceiver Joint Design Algorithm}
\label{HBF_design}
In this section, we first transform problem \eqref{original_problem}
into a more tractable form and then propose a novel iterative
BCD-based algorithm to efficiently solve the converted
problem. Subsequently, we carry out the convergence and computational
complexity analyses for the proposed algorithm.

\subsection{Problem transformation}
Problem \eqref{original_problem} is hard to solve due to the highly
coupled and nonconvex constraints.  Hence, we provide Theorem
\ref{sigma_trans} to simplify problem \eqref{original_problem}.
\begin{theorem}\label{sigma_trans}
	The scaled KKT solution $(\{\mathbf{\bar{P}}_m^{\star},\mathbf{F}_m^{\star}\},\mathbf{T}^{\star},\mathbf{\bar{W}}^{\star},\mathbf{U}^{\star})\triangleq$ $(\{\frac{1}{a} \sigma_m \mathbf{P}_m^{\star},\mathbf{F}_m^{\star}\},\mathbf{T}^{\star},a\mathbf{W}^{\star},\mathbf{U}^{\star})$, noted as $\mathcal{\bar{S}}^{\star}$, of the following problem is a KKT solution of problem \eqref{original_problem},
	\begin{equation}
		\begin{aligned}
		\min_{\mathcal{S}}&\quad \text{MSE}_{\sigma}(\{\mathbf{P}_m,\mathbf{F}_m\},\mathbf{T},\mathbf{W},\mathbf{U})\\
		\text{s.t.} &\quad \eqref{abs_constraint1}, \eqref{abs_constraint2},
		\end{aligned}\label{trans_problem}
	\end{equation}
	where $\mathcal{S}\triangleq\{\{\mathbf{P}_m,\mathbf{F}_m\},\mathbf{T},\mathbf{W},\mathbf{U}\}$, $(\{\mathbf{P}_m^{\star},\mathbf{F}_m^{\star}\},\mathbf{T}^{\star},\mathbf{W}^{\star},\mathbf{U}^{\star})$ is a KKT solution of problem \eqref{trans_problem}, $a=\frac{\sqrt{P_t}}{\|\mathbf{T}^{\star}\mathbf{W}^{\star}\|}$ denotes the scaling factor, and the objective function is given by
	\begin{equation}
\begin{aligned}
	&\text{MSE}_{\sigma}(\{\mathbf{P}_m,\mathbf{F}_m\},\mathbf{T},\mathbf{W},\mathbf{U})\\
	=&\sum_{m=1}^M \text{tr}\bigg(\mathbf{P}_m\mathbf{F}_m\mathbf{\hat{H}}_m\mathbf{TWW}^H\mathbf{T}^H\mathbf{\hat{H}}_m^H\mathbf{F}_m^H\mathbf{P}_m^H
	+\hat{\sigma}_{e,m}^2\text{tr}(\mathbf{TWW}^H\mathbf{T}^H)\mathbf{P}_m\mathbf{F}_m\mathbf{F}_m^H\mathbf{P}_m^H \\
	&+\frac{1}{P_t}\text{tr}(\mathbf{TWW}^H\mathbf{T}^H)\mathbf{P}_m\mathbf{F}_m\mathbf{F}_m^H\mathbf{P}_m^H-\mathbf{P}_m\mathbf{F}_m\mathbf{\hat{H}}_m\mathbf{TWU}^H\mathbf{L}\mathbf{A}_m^H\\
	&-\mathbf{A}_m\mathbf{L}^T\mathbf{U}\mathbf{W}^H\mathbf{T}^H\mathbf{\hat{H}}_m^H\mathbf{F}_m^H\mathbf{P}_m^H\bigg)+\text{tr}(\mathbf{U}\mathbf{U}^H),
\end{aligned}
\end{equation}
where $\mathbf{\hat{H}}_m\triangleq
\frac{\mathbf{\bar{H}}_m}{\sigma_m}$ denotes the scaled channel matrix
and $\hat{\sigma}_{e,m}\triangleq \frac{\sigma_{e,m}}{\sigma_m}$
denotes the scaled channel estimation error variance.
\end{theorem}
\begin{proof}
	See \textbf{Appendix} \ref{proof_sigma_trans}.
\end{proof}

\subsection{Proposed BCD-based joint iterative design}
In this subsection, we propose an efficient BCD-based iterative
algorithm to solve problem \eqref{trans_problem}. The variables are
partitioned into several convenient blocks which are updated
sequentially at each iteration. The subproblems in each block can be
solved in closed form.  At this point, we partition the search
variables into five blocks as follows: 1) Update $\{\mathbf{P}_m\}$ in
parallel, $\forall m\in\mathcal{M}$, by fixing the variables within
the other blocks; 2) Update $\mathbf{U}$ by fixing the other
variables; 3) Update $\{[\mathbf{F}_m]_{i,j}\},\, \forall i,j$,
sequentially by fixing other variables, $\forall m\in\mathcal{M}$. The
unit modulus constraints \eqref{abs_constraint2} are handled within
this block; 4) Update $[\mathbf{T}]_{i,j},\, \forall i,j,$
sequentially by fixing the other variables. The unit modulus
constraints \eqref{abs_constraint1} are handled within this block; 5)
Update $\mathbf{W}$ by fixing the other variables.  The detailed
updating procedure is presented as follows.

In \textbf{Step 1}, the subproblem for $\forall m$ can be expressed as
\begin{equation}
	\min_{\mathbf{P}_m}\quad \hat{f}_m,
\end{equation}
where
\begin{equation}
\begin{aligned}
	\hat{f}_m\triangleq \,&\text{tr}\bigg(\mathbf{P}_m\mathbf{F}_m\mathbf{\hat{H}}_m\mathbf{TWW}^H\mathbf{T}^H\mathbf{\hat{H}}_m^H\mathbf{F}_m^H\mathbf{P}_m^H
	+\hat{\sigma}_{e,m}^2\text{tr}(\mathbf{TWW}^H\mathbf{T}^H)\mathbf{P}_m\mathbf{F}_m\mathbf{F}_m^H\mathbf{P}_m^H\\
	&+\frac{1}{P_t}\text{tr}(\mathbf{TWW}^H\mathbf{T}^H)\mathbf{P}_m\mathbf{F}_m\mathbf{F}_m^H\mathbf{P}_m^H
	-\mathbf{P}_m\mathbf{F}_m\mathbf{\hat{H}}_m\mathbf{TWU}^H\mathbf{L}\mathbf{A}_m^H\\
	&-\mathbf{A}_m\mathbf{L}^T\mathbf{U}\mathbf{W}^H\mathbf{T}^H\mathbf{\hat{H}}_m^H\mathbf{F}_m^H\mathbf{P}_m^H\bigg).
\end{aligned}
\end{equation}
This is an unconstrained convex optimization problem with respect to $\mathbf{P}_m$. By checking the first order optimality condition, we obtain the solution of $\mathbf{P}_m$ as
\begin{equation}
\begin{aligned}
	\mathbf{P}_m^{\star}=\,&\mathbf{A}_m\mathbf{L}^T\mathbf{U}(\mathbf{T}\mathbf{W})^H\mathbf{\hat{H}}_m^H\mathbf{F}_m^H
	\bigg(\mathbf{F}_m\mathbf{\hat{H}}_m\mathbf{T}\mathbf{W}(\mathbf{T}\mathbf{W})^H\mathbf{\hat{H}}_m^H\mathbf{F}_m^H +\hat{\sigma}_{e,m}^2\|\mathbf{T}\mathbf{W}\|^2\mathbf{F}_m\mathbf{F}_m^H\\
	& + \frac{\|\mathbf{TW}\|^2}{P_t}\mathbf{F}_m\mathbf{F}_m^H \bigg)^{-1}. \label{P_solution}
\end{aligned}
\end{equation}

In \textbf{Step 2}, we seek the optimization of matrix $\mathbf{U}$
with the other variables fixed, which minimizes the MSE. This
subproblem is given by
\begin{equation}
\begin{aligned}
	\min_{\mathbf{U}}&\quad \hat{g},
\end{aligned}\label{U_problem}
\end{equation}
where $\hat{g}\triangleq \sum_{m=1}^M
-2\text{Re}\big(\text{tr}(\mathbf{P}_m\mathbf{F}_m\mathbf{\hat{H}}_m\mathbf{TWU}^H\mathbf{L}\mathbf{A}_m^H)\big)+\text{tr}(\mathbf{U}\mathbf{U}^H)$
and $\mathbf{U}$ denotes a lower triangle matrix with ones on the main
diagonal. We can solve this subproblem by taking the partial
derivation of $\hat{g}$ with respect to the elements in the strictly
lower triangle area of the matrix $\mathbf{U}$.

Let us define an operator $\text{vecLT}(\mathbf{X})$, which extracts
the elements in the strictly lower triangle area of the square matrix
$\mathbf{X}\in \mathbb{C}^{n\times n}$ and vectorizes these elements
in the form of column, i.e., $\text{vecLT}(\mathbf{X})=
[[\mathbf{X}]_{2,1}, \ldots, [\mathbf{X}]_{n,1}, [\mathbf{X}]_{3,2},
  \ldots [\mathbf{X}]_{n,2}, [\mathbf{X}]_{4,3} \ldots,
         [\mathbf{X}]_{n,n-1}]^T$.  It is readily seen that
$\frac{\partial\hat{g}}{\partial\text{vecLT}(\mathbf{U})^*}=\text{vecLT}(\frac{\partial\hat{g}}{\partial\mathbf{U}^*})$. Hence,
by checking the first order optimality condition, the optimal strictly
lower triangle part of $\mathbf{U}$ can be given by
\begin{equation}
	\text{vecLT}(\mathbf{U}^{\star})=\text{vecLT}(\sum_{m=1}^M\mathbf{L}\mathbf{A}_m^H\mathbf{P}_m\mathbf{F}_m\mathbf{\hat{H}}_m\mathbf{T}\mathbf{W}),\label{U_solution}
\end{equation}
where $\mathbf{U}^{\star}$ denotes the optimal lower triangle matrix.

In \textbf{Step 3}, we optimize $\mathbf{F}_m, \forall
m\in\mathcal{M}$, in parallel, with the other variables
fixed. Specifically, the elements of $\mathbf{F}_m$, i.e.,
$[\mathbf{F}_m]_{i,j},\, \forall i,j$, are optimized sequentially. The
corresponding subproblem is given by
\begin{subequations}
\begin{align}
	\min_{[\mathbf{F}_m]_{i,j}}&\quad \hat{f}_m\\
	\text{s.t.}&\quad |[\mathbf{F}_m]_{i,j}|=1.\label{F_element_abs}
\end{align}\label{Problem_for_block3}
\end{subequations}

By appropriate rearrangement, we can rewrite problem \eqref{Problem_for_block3} as
\begin{equation}
\begin{aligned}
	\min_{[\mathbf{F}_m]_{i,j}}&\quad\text{tr}(\mathbf{F}_m^H\mathbf{A}_{F_m}\mathbf{F}_m\mathbf{C}_{F_m}-2\text{Re}(\mathbf{F}_m^H\mathbf{B}_{F_m}))\\
	\text{s.t.}&\quad \eqref{F_element_abs},
\end{aligned}\label{one_interation_prob1}
\end{equation}
where $\mathbf{A}_{F_m}\triangleq \mathbf{P}_m^H\mathbf{P}_m$,
$\mathbf{B}_{F_m}\triangleq
\mathbf{P}_m^H\mathbf{A}_m\mathbf{L}^T\mathbf{U}\mathbf{W}^H\mathbf{T}^H\mathbf{\hat{H}}_m^H$
and $\mathbf{C}_{F_m}\triangleq
\mathbf{\hat{H}}_m\mathbf{TWW}^H\mathbf{T}^H\mathbf{\hat{H}}_m^H+\hat{\sigma}_{e,m}^2\text{tr}(\mathbf{TWW}^H\mathbf{T}^H)\mathbf{I}+\frac{1}{P_t}\text{tr}(\mathbf{TWW}^H\mathbf{T}^H)\mathbf{I}$.
It is readily seen that the objective function of problem
\eqref{one_interation_prob1} is a quadratic function with respect to
$[\mathbf{F}_m]_{i,j}$. Thus, by omitting some constant terms in the
objective function, problem \eqref{one_interation_prob1} can be
further rewritten as
\begin{equation}
\begin{aligned}
	\min_{[\mathbf{F}_m]_{i,j}}&\quad \bar{a}_{F,m,i,j}|[\mathbf{F}_m]_{i,j}|^2-\text{Re}(\bar{b}_{F,m,i,j}^*[\mathbf{F}_m]_{i,j})\\
	\text{s.t.}&\quad \eqref{F_element_abs},
\end{aligned}\label{one_interation_prob222}
\end{equation}
where $\bar{a}_{F,m,i,j}$ and $\bar{b}_{F,m,i,j}$ denote some
coefficients.  The optimal solution of problem
\eqref{one_interation_prob222} is given by
$[\mathbf{F}_m]_{i,j}^{\star}=\frac{\bar{b}_{F,m,i,j}}{|\bar{b}_{F,m,i,j}|}$. Therefore,
we only need to know the value of $\bar{b}_{F,m,i,j}$ to update
$[\mathbf{F}_m]_{i,j}$. The value of $\bar{b}_{F,m,i,j}$ is given by
\begin{equation}
	\bar{b}_{F,m,i,j} = [\mathbf{A}_{F,m}]_{i,i}[\mathbf{F}_m]_{i,j}[\mathbf{C}_{F,m}]_{j,j} - [\mathbf{A}_{F,m}\mathbf{F}_m\mathbf{C}_{F,m}]_{i,j} + [\mathbf{B}_{F,m}]_{i,j}.
\end{equation}
Besides, in order to reduce the computational complexity of updating $[\mathbf{F}_m]_{i,j}$, we can update $[\mathbf{F}_m]_{i,j}$ sequentially by following similar steps in \textbf{Algorithm 3} in \cite{shi2018spectral}.

In \textbf{Step 4}, we optimize the elements in $\mathbf{T}$, i.e.,
$[\mathbf{T}]_{i,j},\, \forall i,j$, by fixing the other
variables. The corresponding subproblem is provided as
\begin{equation}
\begin{aligned}
	\min_{[\mathbf{T}]_{i,j}}&\quad \text{MSE}_{\sigma}\\
	s.t. &\quad |[\mathbf{T}]_{i,j}|=1.
\end{aligned}\label{Problem_for_block4}
\end{equation}
This subproblem can be solved by following the same method introduced in \textbf{Step 3}.

In \textbf{Step 5}, we optimize the variable $\mathbf{W}$ with the
other variables fixed. We need to solve the following convex
subproblem
\begin{equation}
\begin{aligned}
	\min_{\mathbf{W}}&\quad \text{MSE}_{\sigma}.
\end{aligned}\label{Problem_for_block5}
\end{equation}

By checking the first order optimality condition, we obtain the
optimal solution of this subproblem as
\begin{equation}
\begin{aligned}
	\mathbf{W}^{\star}=&\bigg(\sum_{m=1}^M\mathbf{T}^H\mathbf{\hat{H}}_m^H\mathbf{F}_m^H\mathbf{P}_m^H\mathbf{P}_m\mathbf{F}_m\mathbf{\hat{H}}_m\mathbf{T}+\text{tr}(\Sigma_{m=1}^M\frac{1}{Pt}\mathbf{P}_m\mathbf{F}_m\mathbf{F}_m^H\mathbf{P}_m^H)\mathbf{T}^H\mathbf{T}\\
	&+\sum_{m=1}^M\hat{\sigma}_{e,m}^2\text{tr}(\mathbf{P}_m\mathbf{F}_m\mathbf{F}_m^H\mathbf{P}_m^H)\mathbf{T}^H\mathbf{T}\bigg)^{-1}
	(\sum_{m=1}^M\mathbf{T}^H\mathbf{\hat{H}}_m^H\mathbf{F}_m^H\mathbf{P}_m^H\mathbf{A}_m\mathbf{L}^T\mathbf{U}).
\end{aligned}\label{W_solution}
\end{equation}

In each iteration of the proposed BCD-based algorithm, we implement
the above five steps to update the optimization variables. The overall
procedure of the proposed algorithm is summarized in \textbf{Algorithm
  \ref{Proposed_algorithm}}.\footnote{We can obtain the KKT point of
  the original problem \eqref{original_problem} via scaling:
  $(\{\mathbf{\bar{P}}_m^{\star},\mathbf{F}_m^{\star}\},\mathbf{T}^{\star},\mathbf{\bar{W}}^{\star},\mathbf{U}^{\star})=(\{\frac{1}{a}
  \sigma_m
  \mathbf{P}_m^{\star},\mathbf{F}_m^{\star}\},\mathbf{T}^{\star},a\mathbf{W}^{\star},\mathbf{U}^{\star})$.}

\begin{algorithm}
\caption{Proposed BCD-based algorithm for nonlinear hybrid transceiver design}
\begin{itemize}
\item[1.] Define the accuracy tolerance $\delta$. Initialize all the variables in $\mathcal{S}$ with a feasible point. Set the iteration number $i=0$.
	\item[2.] \textbf{Repeat}
	\begin{itemize}
		\item[2.1] Update $ \mathbf{P}_m, \forall m\in\mathcal{M}$, in parallel based on \eqref{P_solution}.
		\item[2.2] Update $\mathbf{U}$ based on \eqref{U_solution}.
		\item[2.3] Update $\mathbf{F}_m, \forall m\in\mathcal{M}$, in parallel. In particular, the elements of $\mathbf{F}_m$ are optimized sequentially based on the method introduced in \textbf{Step 3}.
		\item[2.4] Update the elements of $\mathbf{T}$ sequentially based on the method introduced in \textbf{Step 3}.
		\item[2.5] Update $\mathbf{W}$ based on \eqref{W_solution}.
		\item[2.6] Update the iteration number: $i=i+1$.
	\end{itemize}
	\item[3.] \textbf{Until} the difference between two successive objective value is less than $\delta$.	
\end{itemize}

\label{Proposed_algorithm}
\end{algorithm}

\subsection{Convergence and complexity of \textbf{Algorithm \ref{Proposed_algorithm}}}
In this subsection, we analyze the convergence and the computational
complexity of the proposed algorithm for the nonlinear hybrid
transceiver joint design.

It is readily seen that each subproblem of the proposed BCD-based
algorithm (\textbf{Algorithm \ref{Proposed_algorithm}}) is uniquely
and globally solved. Hence the proposed algorithm converges to a KKT
point of problem \eqref{trans_problem} \cite{wright2015coordinate}.

The complexity of the proposed BCD-based algorithm is dominated by the
inversion operations in \textbf{Step 1} and \textbf{Step 5} and the
multiplications in \textbf{Step 3} and \textbf{Step 4}, the
complexities of which are $\mathcal{O}(R_{d,m}^3)$,
$\mathcal{O}(R_s^3)$, $\mathcal{O}(N_{d,m}^2R_{d,m}^2)$ and
$\mathcal{O}(N_s^2R_s^2)$, respectively.  Therefore, by omitting the
lower order terms, the complexity of our proposed BCD-based algorithm
is given by $\mathcal{O}(I(N_s^2R_s^2))$, where $I$ denotes the
maximum iteration number of the proposed BCD-based algorithm.

\section{Proposed two-timescale Hybrid Transceiver joint design algorithm}
\label{TTS_design}
In order to reduce the CSI signalling overhead and the effects of
outdated CSI caused by the associated delays, in this section we
propose a novel two-timescale nonlinear hybrid transceiver design
algorithm.  In this scheme, the long-timescale analog BF matrices are
optimized based on the channel statistics and the short-timescale
digital processing matrices are designed based on the instantaneous
low-dimensional effective CSI matrices.
Based on the TOSCA framework \cite{liu2018online}, we can see that
problem \eqref{TTS_problem} can be decomposed into a long-timescale
master problem and a short-timescale subproblem.

\subsection{Short-timescale subproblem}
By fixing the long-timescale variables $\{\mathbf{F}_m\}$ and $\mathbf{T}$, the short-timescale subproblem is given by
\begin{equation}
		\begin{aligned}
		\min_{\mathbf{W},\mathbf{U},\{\mathbf{P}_m\}}&\quad \text{MSE}(\{\mathbf{P}_m\},\mathbf{W},\mathbf{U})\\
		\text{s.t.} &\quad \eqref{power_constraint10}.
		\end{aligned}\label{short_term_problem}
\end{equation}
Note that Theorem \ref{sigma_trans} can be also applied to this subproblem similarly. This subproblem can be transformed into the following problem,
\begin{equation}
		\begin{aligned}
		\min_{\mathbf{W},\mathbf{U},\{\mathbf{P}_m\}}\quad \text{MSE}_{\sigma}(\{\mathbf{P}_m\},\mathbf{W},\mathbf{U}).
		\end{aligned}\label{short_term_trans_problem}
\end{equation}
We can solve this converted short-timescale problem based on a BCD
algorithm which is similar to \textbf{Algorithm
  \ref{Proposed_algorithm}} (without Step 2.3 and Step 2.4). Then, we
obtain the solution of the short-timescale subproblem
\eqref{short_term_problem} via scaling:
$(\{\mathbf{\bar{P}}_m^{\star}\},\mathbf{\bar{W}}^{\star},\mathbf{U}^{\star})=(\{\frac{1}{a}
\sigma_m
\mathbf{P}_m^{\star}\},a\mathbf{W}^{\star},\mathbf{U}^{\star})$.

\textbf{Remark:} Note that the design of short-timescale digital processing matrices only requires the effective CSI matrices $\mathbf{\tilde{H}}_m$, which can be obtained by pre-multiplying and post-multiplying $\mathbf{\hat{H}}^i_m$ with $\mathbf{F}^t_m$ and $\mathbf{T}^t$, respectively, i.e., $\mathbf{F}_m^t\mathbf{\hat{H}}_m^i\mathbf{T}^t=\mathbf{\tilde{H}}_m^i $. 
The effective CSI matrices $\{\mathbf{\tilde{H}}_m\}$ have much lower dimension than the instantaneous estimated CSI matrices $\{\mathbf{\hat{H}}_m\}$, thus the overhead of sending CSI can be significantly reduced.

\subsection{Long-timescale master problem}
By fixing the short-timescale variables, the long-timescale master problem is given by
\begin{equation}
		\begin{aligned}
		\min_{\boldsymbol{\theta}_T,\{\boldsymbol{\theta}_{F_m}\}}&\quad \tilde{f}(\boldsymbol{\theta}_T,\{\boldsymbol{\theta}_{F_m}\}, \boldsymbol{\Theta}^{\star})=\mathbb{E}_{\mathbf{\bar{H}}_m}\{g(\boldsymbol{\theta}_T,\{\boldsymbol{\theta}_{F_m}\}, \boldsymbol{\Theta}^{\star})\}
		\end{aligned}\label{long_term_problem}
	\end{equation}
where $ \boldsymbol{\theta}_{F_m}\triangleq \angle\mathbf{F}_m , \forall m$, $\boldsymbol{\theta}_T\triangleq \angle \mathbf{T}$, $\boldsymbol{\Theta}^{\star}\triangleq \{\{\mathbf{\bar{P}}_m\}^{\star},\mathbf{\bar{W}}^{\star},\mathbf{U}^{\star}\}$ denotes the solution of problem \eqref{short_term_problem} and
\begin{equation}
	g(\boldsymbol{\theta}_T,\{\boldsymbol{\theta}_{F_m}\}, \boldsymbol{\Theta}^{\star})\triangleq \text{MSE}(\{\mathbf{P}_m,\mathbf{F}_m\},\mathbf{T},\mathbf{W},\mathbf{U}).
\end{equation}

To solve problem \eqref{long_term_problem}, the long-timescale
variables are updated at the end of each frame by solving an
approximation problem obtained via replacing the objective function of
problem \eqref{long_term_problem} with a quadratic surrogate function.
Thus, we introduce the following quadratic surrogate function to
approximate the objective function for frame $t$:
\begin{equation}
\begin{aligned}
	\bar{f}^t(\boldsymbol{\theta}_T,\{\boldsymbol{\theta}_{F_m}\},\boldsymbol{\Theta}^{\star,t})=&f^{t}+(\mathbf{f}_{T}^{t})^T(\boldsymbol{\theta}_T-\boldsymbol{\theta}_T^{t})+\sum_{m=1}^M(\mathbf{f}_{F_m}^{t})^T(\boldsymbol{\theta}_{F_m}-\boldsymbol{\theta}_{F_m}^{t})\\
	&+\tau\|\boldsymbol{\theta}_T-\boldsymbol{\theta}_T^{t}\|^2+\sum_{m=1}^M\tau\|\boldsymbol{\theta}_{F_m}-\boldsymbol{\theta}_{F_m}^{t}\|^2,
\end{aligned}\label{surrogate_func}
\end{equation}
where $\boldsymbol{\Theta}^{\star,t}$ denotes the solution of solving
problem \eqref{short_term_problem} with given
$\{\mathbf{\bar{H}}_m^t\}$, $\boldsymbol{\theta}_T^{t}$ and
$\boldsymbol{\theta}_{F_m}^{t}$. $\tau>0$ is a constant. $f^t$,
$\mathbf{f}_{T}^t$ and $\mathbf{f}_{F_m}^t$ denote the approximations
of objective function $\tilde{f}$, the partial derivatives
$\frac{\partial \tilde{f}}{\partial \boldsymbol{\theta}_T}$ and
$\frac{\partial \tilde{f}}{\partial \boldsymbol{\theta}_{F_m}}$,
respectively, based on the current channel sample
$\{\mathbf{\bar{H}}_m^t\}$ and $\boldsymbol{\Theta}^{\star, t}$. The
quantities can be updated based on the following expressions:
\begin{equation}
	f^t=(1-\rho^t)f^{t-1}+\rho^t g(\boldsymbol{\theta}_T^{t},\{\boldsymbol{\theta}_{F_m}^{t}\},\boldsymbol{\Theta}^{\star, t}),
\end{equation}
\begin{equation}
	\mathbf{f}_{T}^{t}=(1-\rho^t)\mathbf{f}_{T}^{t-1}+\rho^t\frac{\partial g}{\partial \boldsymbol{\theta}_T}|_{(\boldsymbol{\theta}_T^{t},\{\boldsymbol{\theta}_{F_m}^{t}\},\boldsymbol{\Theta}^{\star, t})},
\end{equation}
and
\begin{equation}
	\mathbf{f}_{F_m}^{t}=(1-\rho^t)\mathbf{f}_{F_m}^{t-1}+\rho^t\frac{\partial g}{\partial \boldsymbol{\theta}_{F_m}}|_{(\boldsymbol{\theta}_T^{t},\{\boldsymbol{\theta}_{F_m}^{t}\},\boldsymbol{\Theta}^{\star, t})}.
\end{equation}
 The details of the derivatives are given in \textbf{Appendix
   \ref{derivation_of_gradients}}. Here $\{\rho^t\}$ is a sequence of
 parameters satisfying condition \eqref{parameter_condition}.

Subsequently, let us solve the approximated problem for \eqref{long_term_problem}, which is given by 
\begin{equation}
	\min_{\boldsymbol{\theta}_T, \{\boldsymbol{\theta}_{F_m}\}} \bar{f}^t(\boldsymbol{\theta}_T, \{\boldsymbol{\theta}_{F_m}\}, \boldsymbol{\Theta}^{\star,t}).\label{approximate_prob}
\end{equation}
It is readily seen that \eqref{approximate_prob} can be solved as follows
\begin{equation}
\begin{aligned}
	&\boldsymbol{\bar{\theta}}_T = \boldsymbol{\theta}^{t}-\frac{\mathbf{f}_T^t}{2\tau}, \quad\boldsymbol{\bar{\theta}}_{F_m} = \boldsymbol{\theta}_{F_m}^{t} - \frac{\mathbf{f}_{F_m}^t}{2\tau}, \forall m,
\end{aligned}
\end{equation}
where $\boldsymbol{\bar{\theta}}_T$ and $\{\boldsymbol{\bar{\theta}}_{F_m}\}$ are the optimal solution of the quadratic approximation problem \eqref{approximate_prob}.

Then, the long-timescale variables are updated as
\begin{equation}\label{update_rule}
	\boldsymbol{\theta}_{T}^{t+1} = (1-\gamma^t)\boldsymbol{\theta}_T^{t} + \gamma^{t}\boldsymbol{\bar{\theta}}_T, \quad \boldsymbol{\theta}_{F_m}^{t+1} = (1-\gamma^t)\boldsymbol{\theta}_{F_m}^{t} + \gamma^{t}\boldsymbol{\bar{\theta}}_{F_m}, \forall m,
\end{equation}
where $\{\gamma^t\}$ denotes a sequence of parameters satisfying condition \eqref{parameter_condition}. The proposed two-timescale joint design algorithm is summarized in \textbf{Algorithm \ref{TTS_algorithm}}.

\begin{algorithm}
\caption{Proposed TOSCA-based two-timescale joint design algorithm}
\begin{itemize}
\item[1.]A super-frame starts. Initialize the long-timescale variables $\{\boldsymbol{\theta}_T^0,\{\boldsymbol{\theta}_{F_m}^0\}\}$ and the short-timescale variables $\{\mathbf{W}^0,\mathbf{U}^0, \{\mathbf{P}_m^0\}\}$ to a feasible point. Set the frame index $t=0$ and the time slot index $i=0$.
	\item[2.] \textbf{Repeat}
	\begin{itemize}
		\item[2.1] Obtain the effective CSI matrices $\{\mathbf{\tilde{H}}_m^i\}$ for time slot $i$.
		\item[2.2] Solve problem \eqref{short_term_problem} and obtain the solution $\{\mathbf{\bar{W}}^i,\mathbf{U}^i, \{\mathbf{\bar{P}}_m^i\}\}$.
		\item[2.3] Update the time slot index: $i=i+1$.
	\end{itemize}
	\item[] \textbf{Until} the frame ends, i.e. $i=(t+1)T_s$.	
	\item[3.] Obtain a CSI sample $\{\mathbf{\bar{H}}_m^t\}$ at the end of frame $t$.
	\item[4.] Update the surrogate function \eqref{surrogate_func} using $\boldsymbol{\Theta}^{\star,t}$, $\{\boldsymbol{\theta}_T^{t}, \{\boldsymbol{\theta}_{F_m}^{t}\}\}$ and $\{\mathbf{\bar{H}}_m^t\}$.
	\item[5.] Solve \eqref{approximate_prob} to obtain $\{\boldsymbol{\bar{\theta}}_T, \boldsymbol{\bar{\theta}}_{F_m}\}$.
	\item[6.] Update $\{\boldsymbol{\theta}_T^{t+1}, \{\boldsymbol{\theta}_{F_m}^{t+1}\}\}$ according to \eqref{update_rule}.
	\item[7.] Set $t=t+1$ and return to Step 2.
\end{itemize}

\label{TTS_algorithm}
\end{algorithm}

According to \cite{liu2018online}, if we choose the sequences of the parameters $\{\rho^t, \gamma^t\}$ so that they satisfy the following condition
\begin{equation}
\begin{aligned}\label{parameter_condition}
&\rho^t\to 0, \frac{1}{\rho^t}\leq O(t^{\beta})\text{ for some }\beta\in(0,1), \sum_t(\rho^t)^2<\infty,\\
&\gamma^t\to 0, \sum_t \gamma^t=\infty, \sum_t (\gamma^t)^2<\infty, \lim_{t\to \infty}\frac{\gamma^t}{\rho^t}=0.
\end{aligned}
\end{equation}
then our proposed two-timescale algorithm can be guaranteed to converge to a KKT solution of problem \eqref{TTS_problem}.
The overall computational complexity of \textbf{Algorithm
  \ref{TTS_algorithm}} is dominated by the updating of the
short-timescale variables, which is given by
$\mathcal{O}(T_sI(R_s^3))$, where $I$ denotes the maximum iteration
number of the proposed BCD-based algorithm for the short-timescale
subproblem.

\section{Design of cancellation ordering matrix}
\label{cancellation_order}

The cancellation order of the proposed THP-based nonlinear hybrid
transceiver design affects the system performance. In order to further
increase the performance, in this section we seek to design the
near-optimal cancellation ordering matrix $\mathbf{L}$.\footnote{In
  this work, we mainly focus on the cancellation order among the users
  rather than that among the data streams, since the channel conditions
  related to different antennas per user are quite similar. }

Let us rewrite the expression of the original $MSE_{\sigma}$ as follows
\begin{equation}
\begin{aligned}
	MSE_{\sigma}(\mathcal{P},\mathbf{U};\mathbf{L}) = g(\mathcal{P};\mathbf{L}) + f(\mathcal{P},\mathbf{C};\mathbf{L})
\end{aligned}\label{MSE_S}
\end{equation}
where $\mathcal{P}\triangleq \{ \{\mathbf{P}_m,\mathbf{F}_m\}, \mathbf{T}, \mathbf{\tilde{W}} \}$ denotes a set of variables,
\begin{equation}
	f(\mathcal{P},\mathbf{C};\mathbf{L})\triangleq \sum_{m=1}^M\text{tr}\bigg(-\mathbf{P}_m\mathbf{F}_m\hat{\mathbf{H}}_m\mathbf{T\tilde{W}}\mathbf{L}^T\mathbf{C}^H\mathbf{L}\mathbf{A}_m^H-\mathbf{A}_m\mathbf{L}^T\mathbf{C}\mathbf{L}\mathbf{\tilde{W}}^H\mathbf{T}^H\mathbf{\hat{H}}_m^H\mathbf{F}_m^H\mathbf{P}_m^H\bigg)+\text{tr}(\mathbf{C}\mathbf{C}^H),\label{MSE_f}
\end{equation}
and
\begin{equation}
\begin{aligned}
	g(\mathcal{P};\mathbf{L})\triangleq& \|\mathbf{L}\mathbf{\bar{A}}\mathbf{T}\mathbf{\tilde{W}}\mathbf{L}^T\|^2+\tilde{g}(\mathcal{P})=\|\mathbf{\bar{A}}\mathbf{T}\mathbf{\tilde{W}}\|^2+\tilde{g}(\mathcal{P}),
\end{aligned}\label{MSE_g}
\end{equation}
where $\mathbf{\bar{A}}\triangleq [ (\mathbf{P}_1\mathbf{F}_1\mathbf{\hat{H}}_1)^T,\,
(\mathbf{P}_2\mathbf{F}_2\mathbf{\hat{H}}_2)^T,\ldots
,(\mathbf{P}_M\mathbf{F}_M\mathbf{\hat{H}}_M)^T
]^T$, $\mathbf{\tilde{W}}=\mathbf{W}\mathbf{L}$ and
\begin{equation}
\begin{aligned}
	\tilde{g}(\mathcal{P})\triangleq& \sum_{m=1}^M \text{tr}\bigg(\hat{\sigma}_{e,m}^2\text{tr}(\mathbf{T\tilde{W}\tilde{W}}^H\mathbf{T}^H)\mathbf{P}_m\mathbf{F}_m\mathbf{F}_m^H\mathbf{P}_m^H
	+\text{tr}(\mathbf{T\tilde{W}\tilde{W}}^H\mathbf{T}^H)\mathbf{P}_m\mathbf{F}_m\mathbf{F}_m^H\mathbf{P}_m^H\\
	&-\mathbf{P}_m\mathbf{F}_m\hat{\mathbf{H}}_m\mathbf{T\tilde{W}}\mathbf{A}_m^H-\mathbf{A}_m\mathbf{\tilde{W}}^H\mathbf{T}^H\mathbf{\hat{H}}_m^H\mathbf{F}_m^H\mathbf{P}_m^H\bigg) + D.
\end{aligned}
\end{equation}

By recalling the solution of $\mathbf{U}$ in \eqref{U_solution} and substituting $\mathbf{C}^{\star}=\mathbf{U}^{\star}-\mathbf{I}$ into $f(\mathcal{P},\mathbf{U};\mathbf{L})$, we obtain
\begin{equation}\label{gain_of_THP}
	f(\mathcal{P},\mathbf{C}^{\star};\mathbf{L})=-\|\Delta(\mathbf{L}\mathbf{\bar{A}}\mathbf{T}\mathbf{\tilde{W}}\mathbf{L}^T)\|^2,
\end{equation}
where the operation $\Delta(.)$ is defined as that
$\Delta(\mathbf{X})$ extracts the elements in the strictly lower
triangle area of the square matrix $\mathbf{X}\in \mathbb{C}^{n\times
  n}$ and forms a strictly lower triangle matrix, i.e.
\begin{equation}
\begin{aligned}
	\left[ \Delta(\mathbf{X}) \right]_{i,j}=\begin{cases}[\mathbf{X}]_{i,j}, \quad & \text{if } i>j,\\
	0, \quad &\text{otherwise}.
	\end{cases}
\end{aligned}
\end{equation}

Based on \eqref{gain_of_THP}, \eqref{MSE_g} and \eqref{MSE_S}, it is
readily seen that $g(\mathcal{P};\mathbf{L})$ is an expression of MSE
of a linear hybrid transceiver, and it is slightly affected by the
cancellation ordering matrix $\mathbf{L}$. The MSE gain of THP comes
from $f(\mathcal{P},\mathbf{C}^{\star};\mathbf{L})$.  Thus, we can see
that the MSE performance of the THP-based hybrid transceiver always
outperforms that of its linear counterpart no matter what the
cancellation ordering matrix $\mathbf{L}$ is.  Note that the matrix
$\Delta(\mathbf{L}\mathbf{\bar{A}}\mathbf{T}\mathbf{\tilde{W}}\mathbf{L}^T)$
is a strictly lower triangle matrix and the elements in its upper
triangle area are forced to zero. The permutation matrix $\mathbf{L}$
can change the positions of the elements and thus the value of
$MSE_{\sigma}(\mathcal{P},\mathbf{U}^{\star};\mathbf{L})$.

However, it is very difficult to design the optimal cancellation
ordering matrix $\mathbf{L}$ due to the MSE expression with unknown
optimization variables.  Thus, we seek to develop a low-complexity
approach to find the near-optimal cancellation order based on a
comparable lower bound of MSE, which does not contain the coupled
terms of $\mathbf{L}$ and other variables.  To this end, we derive a lower bound for the term $f(\mathcal{P},\mathbf{C}^{\star};\mathbf{L})$ and
design the matrix $\mathbf{L}$ based on the lower bound.


%

In the following, we derive the comparable lower bound of
$f(\mathcal{P},\mathbf{C}^{\star};\mathbf{L})$.  First, let us define
$\mathbf{\bar{B}}\triangleq \mathbf{T}\mathbf{W} =
[\mathbf{\bar{B}}_1, \mathbf{\bar{B}}_2, ...,
  \mathbf{\bar{B}}_M]\in\mathbb{C}^{N_s\times D}$, where
$\mathbf{\bar{B}}_m\in\mathbb{C}^{N_s\times D_m}$ denotes a submatrix
of $\mathbf{\bar{B}}$, which is formulated from the
$(\sum_{i=1}^{m-1}D_i+1)$th column vector to the
$(\sum_{i=1}^{m}D_i)$th column vector of matrix $\mathbf{\bar{B}}$,
then we have $f(\mathcal{P},\mathbf{C}^{\star};\mathbf{L}) =
-\|\Delta(\mathbf{L}\mathbf{\bar{A}}\mathbf{\bar{B}})\|^2 =
-\|\Delta(\mathbf{\hat{A}}\mathbf{\bar{B}})\|^2 =
-\|\Delta(\boldsymbol{\Omega})\|^2$, where $\mathbf{\hat{A}}
\triangleq \mathbf{L}\mathbf{\bar{A}}=[
  (\mathbf{P}_{(1)}^{\star}\mathbf{F}_{(1)}^{\star}\mathbf{\hat{H}}_{(1)})^T,\,
  (\mathbf{P}_{(2)}^{\star}\mathbf{F}_{(2)}^{\star}\mathbf{\hat{H}}_{(2)})^T$,
  $\ldots
  ,(\mathbf{P}_{(M)}^{\star}\mathbf{F}_{(M)}^{\star}\mathbf{\hat{H}}_{(M)})^T
]^T \in \mathbb{C}^{D\times N_s}$ denotes a matrix obtained by
permutating the rows of $\mathbf{\bar{A}}$ with the ordering matrix
$\mathbf{L}$.  $\boldsymbol{\Omega}=\mathbf{\hat{A}} \mathbf{\bar{B}}$
can be structured as
\begin{equation}
\boldsymbol{\Omega}\triangleq
\left[\begin{array}{cccc}
\boldsymbol{\Omega}_{(1),1} & \boldsymbol{\Omega}_{(1),2} & \ldots & \boldsymbol{\Omega}_{(1),M}\\
\boldsymbol{\Omega}_{(2),1} & \boldsymbol{\Omega}_{(2),2} & \ldots & \boldsymbol{\Omega}_{(2),M}\\
\vdots & \vdots & \ddots & \vdots \\
\boldsymbol{\Omega}_{(M),1} & \boldsymbol{\Omega}_{(M),2} & \ldots & \boldsymbol{\Omega}_{(M),M}\\
\end{array}\right],
\end{equation}
where $\boldsymbol{\Omega}_{(i),k}\triangleq \mathbf{P}_{(i)}^{\star}\mathbf{F}_{(i)}^{\star}\mathbf{\hat{H}}_{(i)}\mathbf{\bar{B}}_k\in \mathbb{C}^{D_{(i)}\times D_{k}}, \forall i,k$.

Then we can rewrite $f(\mathcal{P},\mathbf{C}^{\star};\mathbf{L})$ as follows
\begin{equation}
	f(\mathcal{P},\mathbf{C}^{\star};\mathbf{L})=-\sum_{i>k}\|\boldsymbol{\Omega}_{(i),k}\|^2 - \sum_{i=1}^M\|\Delta(\boldsymbol{\Omega}_{(i),i})\|^2.\label{f_cal}
\end{equation}

Let us define $C_1$ as the upper bound of $\|\mathbf{P}_m\mathbf{F}_m\|\| \mathbf{\bar{B}}_k\|$, i.e., $\|\mathbf{P}_m\mathbf{F}_m\|\| \mathbf{\bar{B}}_k\|\leq C_1, \forall m, k$. Then, we have the upper bound for $\|\boldsymbol{\Omega}_{(i),k}\|^2$ as
\begin{equation}
	\|\boldsymbol{\Omega}_{(i),k}\|^2\leq (\|\mathbf{P}_{(i)}^{\star}\mathbf{F}_{(i)}^{\star}\|\|\mathbf{\hat{H}}_{(i)}\|\|\mathbf{\bar{B}}_k\|)^2\leq C_1^2\|\mathbf{\hat{H}}_{(i)}\|^2.\label{each_ineq}
\end{equation}

Based on \eqref{f_cal} and \eqref{each_ineq}, we finally obtain the lower bound for $f(\mathcal{P},\mathbf{C}^{\star};\mathbf{L})$ as
\begin{equation}
\begin{aligned}
	f(\mathcal{P},\mathbf{C}^{\star};\mathbf{L})&\geq -\sum_{i>k}\|\boldsymbol{\Omega}_{(i),k}\|^2 - \sum_{i}\|\boldsymbol{\Omega}_{(i),i}\|^2\geq -C_1^2(\sum_{i=1}^M i\|\mathbf{\hat{H}}_{(i)}\|^2).
\end{aligned}
\end{equation}
Then, we can see that the cancellation order can be generated from the
smallest value to the largest value based on the sequence
$\|\mathbf{\hat{H}}_{1}\|^2, \|\mathbf{\hat{H}}_{2}\|^2, \ldots,
\|\mathbf{\hat{H}}_{M}\|^2$ aiming at minimizing the lower bound of
MSE, i.e.,
$\|\mathbf{\bar{A}}\mathbf{T}\mathbf{\tilde{W}}\|^2+\tilde{g}(\mathcal{P})-C_1^2(\sum_{i=1}^M
i\|\mathbf{\hat{H}}_{(i)}\|^2)$.  The cancellation ordering matrix
$\mathbf{L}$ can be straightforwardly formulated based on this order.

\section{Simulation results}
\label{simulation_results}
In this section, we evaluate the performance of the proposed THP-based
hybrid transceiver joint design algorithms.  We consider the widely
used narrow-band mmWave channel model with the uniform linear antenna
array configuration\cite{hemadeh2018millimeter}. The channel matrix
between user $m$ and the BS is given by
\begin{equation}
\mathbf{H}_{k}=\sum_{n_{cl}}^{N_{cl}}\sum_{p}^{N_p}\Gamma_{n_{cl}}\mathbf{a}(\theta_{n_{cl}}^t+\psi_{n_{cl},n_p}^t)\mathbf{a}(\theta_{n_{cl}}^r+\psi_{n_{cl},n_p}^r)\text{exp}(j2\pi f_d\tau \text{cos}(\theta_{n_{cl}}^r+\psi_{n_{cl},n_p}^r)),
\end{equation}
where $N_{cl}$ and $N_p$ are the number of aggregated clusters and the
number of rays within the cluster $p$, respectively, and
$\Gamma_{n_{cl}}\sim\mathcal{CN}(0,1)$ represents the complex channel
gain of the cluster $n_{cl}$. $\theta_{n_{cl}}^t$ and
$\psi_{n_{cl},n_p}^t$ are the $n_{cl}$-th cluster's central angle of
departure and bias angles of departure due to the angle spread,
correspondingly, while $\theta_{n_{cl}}^r$ and $\psi_{n_{cl},n_p}^r$
are the $n_{cl}$-th cluster's counterparts of the angles of
arrival. $f_d$ is the maximum Doppler shift, $\tau$ is the delay, and
$\mathbf{a}(\theta)$ is the array response vector whose generic
expression can be given by
\begin{equation}
	\mathbf{a}(\theta)=\frac{1}{\sqrt{N}}[1,e^{jk_od_a\pi \text{sin}(\theta)},...,e^{jk_od_a\pi (N-1)\text{sin}(\theta)}]^T,
\end{equation}
where $k_o =2\pi/\lambda_o$, $\lambda_o$ is the wavelength at the
operating frequency and $d_a$ is the antenna spacing. We assume that
there are 3 clusters and 5 rays within each cluster, i.e., totally 15
rays as in\cite{shi2018spectral}.  Besides, we limit
$\boldsymbol{\theta}_{n_{cl}}^t$ and $\boldsymbol{\theta}_{n_{cl}}^t$
in a range of $(-\frac{\pi}{8},\frac{\pi}{8})$ and set, unless
specified, $\sigma_{e,m}=\sigma_e=0.1,\, \forall m$.

We assume that there are $M=4$ users and each user is equipped with
$N_{d,m}=8$ antennas and $R_{d,m}=2$ RF chains, while the BS has
$N_s=32$ antennas and $R_s=8$ RF chains. We employ the $16$-QAM
modulation.  The number of data streams for each user is set to
$D_m=2,\forall m$, hence the number of data streams at the BS is
$D=MD_m=8$.  The level of noise variance is normalized to
$\sigma_m=1,\forall m$.  The signal-to-noise ratio (SNR) is defined as
$SNR=10\text{log}_{10}\frac{P_t}{\sigma_m^2}\text{dB}$.  We consider
the following algorithms for comparison:
\begin{itemize}
\item \textbf{Nonlinear joint}: The proposed BCD-based algorithm  (\textbf{Algorithm 1}) for the THP-based joint hybrid transceiver design.


\item \textbf{Nonlinear separate}: The analog BF matrices are first obtained by using the channel matching approach as in \cite{cai2019robust}. The THP-based digital processing matrices are optimized jointly.

\item \textbf{FD}: The proposed THP-based fully digital transceiver design algorithm.

\item \textbf{Linear joint}: The joint linear hybrid transceiver design algorithm proposed in \cite{shi2018spectral}.

\item \textbf{Linear separate}: The analog BF matrices are first obtained by using the channel matching approach as in \cite{cai2019robust}. The linear transceiver matrices are optimized jointly.

\item \textbf{ZF}: The analog BF matrices are first obtained by using the channel matching approach as in \cite{cai2019robust}. The digital transceiver matrices are designed based on the conventional zero-forcing (ZF) BF.

\item \textbf{Two-timescale joint}: The proposed TOSCA-based two-timescale joint design algorithm.
\end{itemize}

\textbf{Remark}: Except the TOSCA-based two-timescale joint design algorithm, all the analyzed designs are single-timescale algorithms.

\subsection{Single-timescale joint design algorithm}
\label{simu1}
\begin{figure}[!t]
\centering
\scalebox{0.65}{\includegraphics{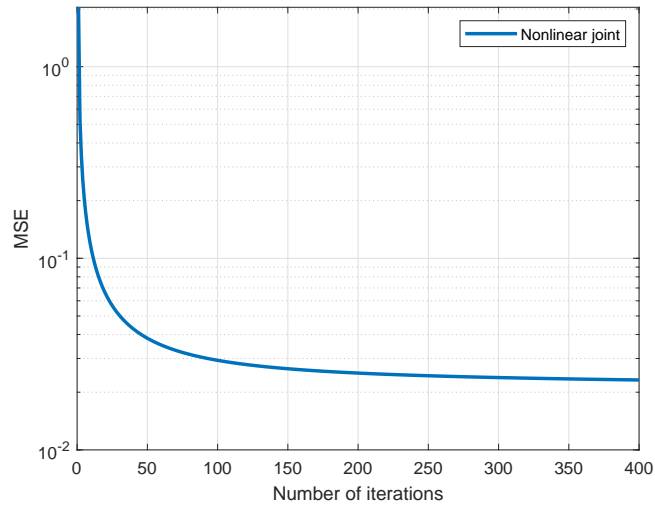}}
\caption{Convergence performance of the proposed algorithm ($SNR=20dB$).}\label{fig:simu1}
\end{figure}
We first study the convergence performance of this proposed
\textbf{Algorithm \ref{Proposed_algorithm}}.  Fig. \ref{fig:simu1}
shows the MSE performance versus the number of iterations, where the
$SNR$ is set to $20dB$. It can be observed that the objective value of
the optimization problem nearly converges within less than 300
iterations, which indicates the convergence behaviour of the proposed
BCD-based joint design algorithm. Moreover, the proposed algorithm
provides relatively low complexity due to the closed-form solutions in
each block.

\begin{figure}[!t]
\centering
\scalebox{0.6}{\includegraphics{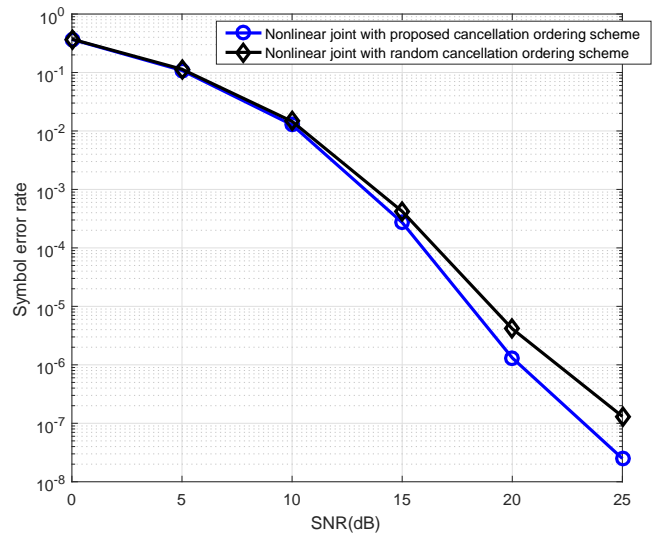}}
\caption{SER performance comparison of the proposed nonlinear hybrid transceiver design algorithm with different cancellation ordering schemes.}\label{fig:simu2}
\end{figure}
Then, we investigate the effect of different cancellation ordering schemes on the symbol error rate (SER) performance.
The proposed nonlinear transceiver design algorithm with the proposed
cancellation ordering scheme is compared with that with the random
cancellation ordering scheme\footnote{The cancellation order is
  generated randomly in this case.}.  Fig. \ref{fig:simu2} shows the
SER performance of the analyzed algorithms.  The results indicate that
the algorithm with the proposed cancellation ordering scheme provides
an almost $3dB$ gain at the SER level of $2\times 10^{-7}$ compared to
the one with the random cancellation ordering scheme, in particular
for the high SNR region, which verifies the effectiveness of the
proposed ordering scheme.

\begin{figure}[!t]
\centering
\scalebox{0.65}{\includegraphics{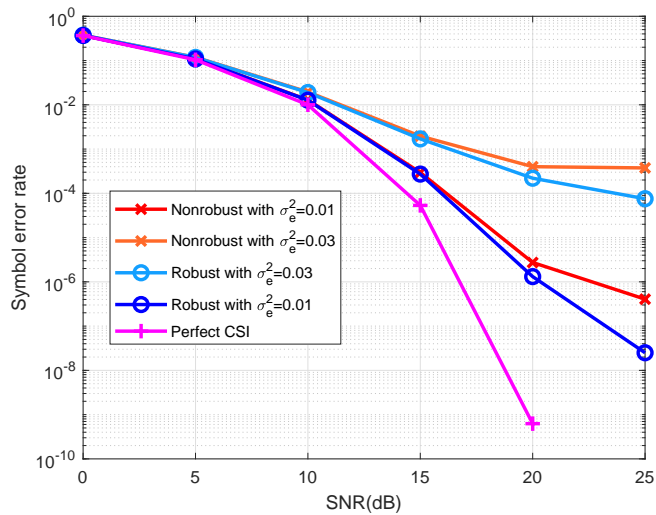}}
\caption{SER performance comparison of the proposed nonlinear hybrid transceiver design algorithm with different variances of CSI errors.}\label{fig:simu3}
\end{figure}

Next, we show the SER performance of the proposed BCD-based joint
design algorithm in the presence of different variances of CSI
errors. As shown in Fig. \ref{fig:simu3}, a smaller variance of CSI
errors leads to better SER performance for the proposed robust
nonlinear transceiver design algorithm and its nonrobust
counterpart.\footnote{The nonrobust transceiver design algorithm
  updates the optimization variables only based on the estimated CSI
  matrices $\{\mathbf{\bar{H}}_m\}$ without considering the channel
  estimation errors.} With the same channel estimation error variance,
the SER performance of the proposed robust design algorithm is always
better than that of the nonrobust design algorithm. Furthermore, the
gap of the SER performance between the robust and nonrobust design
algorithms increases as the increasing of SNR. The results verify the
robustness of the proposed nonlinear transceiver design algorithm.

\begin{figure}[!t]
\centering
\scalebox{0.65}{\includegraphics{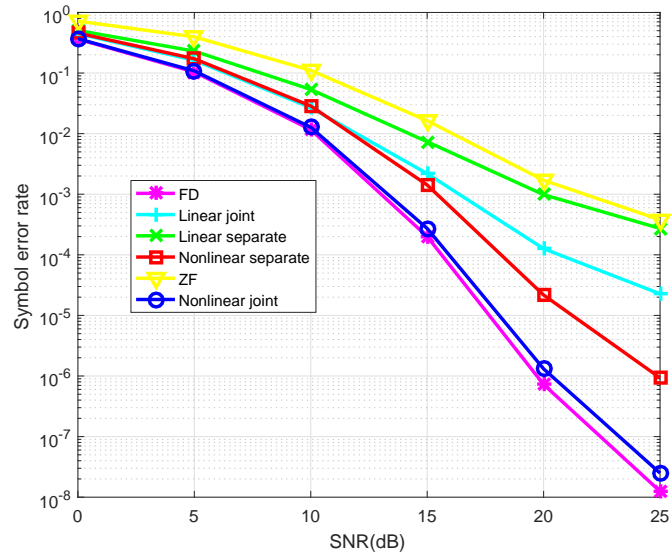}}
\caption{SER performance comparison for different transceiver design algorithms.}\label{fig:simu4}
\end{figure}

\begin{figure}[!t]
\centering
\scalebox{0.65}{\includegraphics{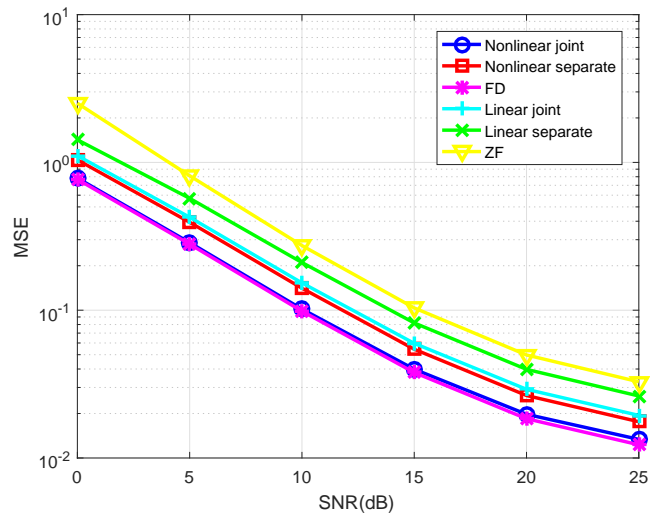}}
\caption{MSE performance comparison for different transceiver design algorithms.}\label{fig:simu5}
\end{figure}

Fig. \ref{fig:simu4} shows the SER performance versus SNR for
different transceiver design algorithms, including the ZF algorithm,
linear joint design algorithm, proposed nonlinear joint design
algorithm, linear separate design algorithm and nonlinear separate
design algorithm.  We observe that, as expected, the nonlinear
transceiver design algorithms provide better SER performance compared
to the linear transceiver design algorithms all the time due to the
successive interference suppression based preprocessing.  Besides, the
linear and nonlinear joint design algorithms significantly outperform
the linear and nonlinear separate design algorithms, respectively, due
to the joint optimization techniques.  Among the hybrid transceiver
design algorithms, the best performance is achieved by the proposed
nonlinear transceiver joint design algorithm followed by the nonlinear
transceiver separate design algorithm, the linear transceiver joint
design algorithm, the linear transceiver separate design algorithm and
the ZF algorithm.  The performance of the FD nonlinear transceiver
design algorithm is provided as a reference. We can see that the
proposed hybrid algorithm can approach the performance of the
performance of the FD transceiver design algorithm.  The corresponding
MSE performance is shown in Fig. \ref{fig:simu5} which coincides with
the results in Fig. \ref{fig:simu4}.

\subsection{Two-timescale joint design algorithm}
\label{simu2}
\begin{figure}[!t]
\centering
\scalebox{0.65}{\includegraphics{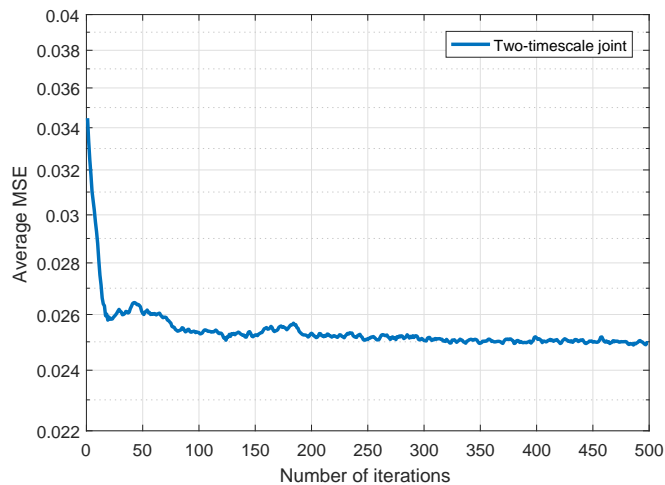}}
\caption{Convergence performance of the proposed two-timescale joint design algorithm $(SNR=20dB)$.}\label{fig:simu6}
\end{figure}

We illustrate the performance of the proposed two-timescale joint design algorithm (\textbf{Algorithm \ref{TTS_algorithm}}).
We assume that the CSI delay is proportional to the dimension of the channel matrices used to update the THP digital processing matrices and the analog BF matrices as in \cite{liu2016impact}. Hence we have
\begin{equation}
\frac{\tau}{\tau_{TTS}}=\frac{N_s\sum_{m=1}^MN_{d,m}}{R_s\sum_{m=1}^MR_{d,m}},\label{frac_TTS}
\end{equation}
where $\tau$ is the full CSI delay of the single-timescale algorithm
and $\tau_{TTS}$ is the effective CSI delay of the two-timescale
algorithm.  We first study the convergence performance of the proposed
two-timescale joint design algorithm under the setting $SNR = 20dB$,
$\sigma_e^2 = 0.01$ and $\tau = 1ms$.  $\tau_{TTS}$ can be computed
based on \eqref{frac_TTS}, which is given by $0.0625ms$. For
simplicity, we omit the computation of $\tau_{TTS}$ in the following
experiments.  Fig. \ref{fig:simu6} illustrates the average MSE versus
the number of iterations. We can see that the proposed two-timescale
algorithm converges within 300 iterations.

\begin{figure}[!t]
\centering
\scalebox{0.57}{\includegraphics{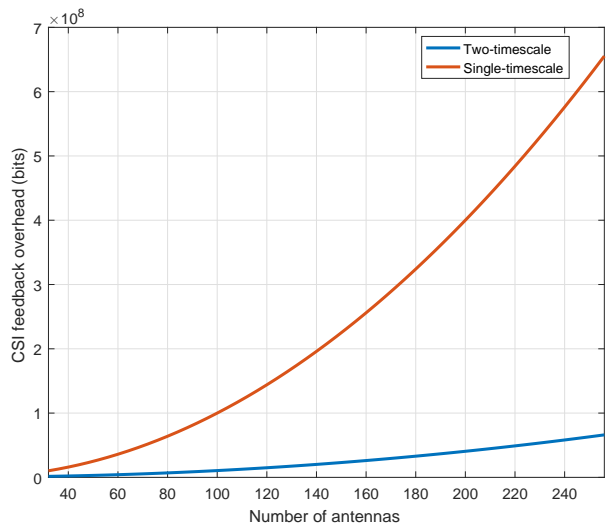}}
\caption{Comparison of the CSI feedback overhead between the single-timescale design algorithm and the two-timescale joint design algorithm.}\label{fig:simu7}
\end{figure}

Then, let us compare the CSI feedback overhead for the
single-timescale algorithm and the two-timescale algorithm.  Let $B$
denote the number of quantization bits needed for each element of the
CSI matrices.  Then the expression of CSI feedback overhead for the
two-timescale algorithm in a super-frame is given by
$T_fN_s\sum_{m=1}^MN_{d,m}+T_f(T_s-1)R_s\sum_{m=1}^MR_{d,m}$.
Similarly, we can obtain the counterpart of the single-timescale
algorithm in a super-frame as
$T_fT_sN_s\sum_{m=1}^MN_{d,m}$. Fig. \ref{fig:simu7} shows the CSI
feedback overhead of the two-timescale and single-timescale algorithms
where $T_f=1000$, $T_s=10$ and $N_s=\sum_{m=1}^MN_{d,m}=N_a$. $N_a$
denotes the number of antennas. This is consistent with the LTE
standard \cite{viering2002spatial}, i.e., the channel statistics
coherence time and the channel coherence time are 10s and 1ms,
respectively.  We can conclude from Fig. \ref{fig:simu7} that the
two-timescale algorithm has a significantly lower CSI feedback
overhead than the single-timescale algorithm.

\begin{figure}[!t]
\centering
\scalebox{0.65}{\includegraphics{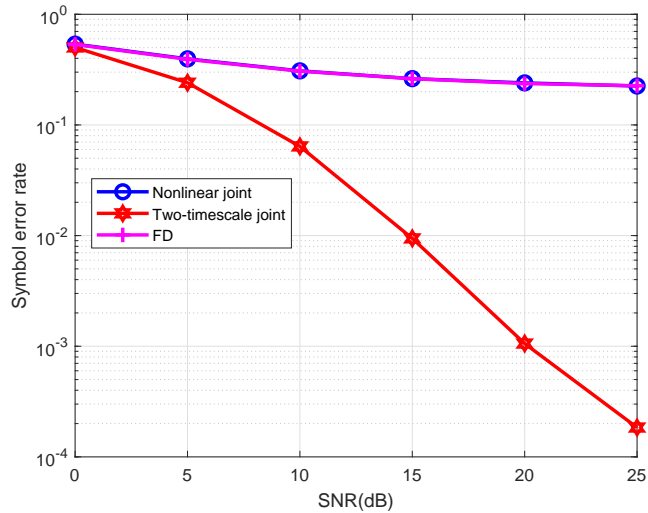}}
\caption{SER performance comparison for different analyzed algorithms under the CSI delay $\tau=5ms$.}\label{fig:simu8}
\end{figure}

Fig. \ref{fig:simu8} shows the SER performance for different analyzed
transceiver design algorithms under the CSI delay $\tau=5ms$,
including the proposed single-timescale joint design algorithm, the
proposed two-timescale joint design algorithm and the FD algorithm.
The two-timescale joint design algorithm provides the best SER
performance while the other algorithms provide almost the same SER
performance.  With the increasing of SNR, the performance gap between
the two-timescale algorithm and the single-timescale algorithm becomes
larger.

In Fig. \ref{fig:simu9}, we show the SER performance of the proposed
single-timescale joint design algorithm, the proposed two-timescale
joint design algorithm and the FD algorithm versus the CSI delay
$\tau$.  It can be observed from Fig. \ref{fig:simu9} that with the
increasing of $\tau$, the performance of the single-timescale
algorithms degrades dramatically while the performance of the
two-timescale algorithm varies slightly.  The two-timescale algorithm
starts to outperform the single-timescale algorithms at the delay of
2ms.  This is mainly because the two-timescale algorithm has much
lower feedback overhead and therefore it creates much smaller CSI
delay.  Moreover, the proposed single-timescale algorithm provides
better performance compared to the proposed two-timescale algorithm in
the presence of smaller delays.  The results verify the robustness of
our proposed TOSCA-based two-timescale joint design algorithm.  This
verifies the effectiveness of the proposed two-timescale algorithm
against the CSI mismatch caused by the delays.

\begin{figure}[!t]
\centering
\scalebox{0.65}{\includegraphics{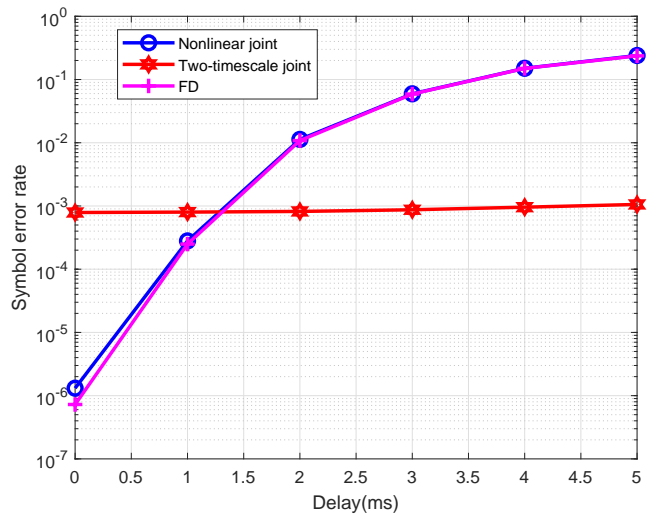}}
\caption{SER performance of different analyzed algorithms versus $\tau$ $(SNR=20 dB)$.}\label{fig:simu9}
\end{figure}

\section{Conclusion}
\label{conclusion}
In this work, we conceived the THP-based joint nonlinear hybrid A/D
transceiver design algorithms for the downlink multiuser MIMO mmWave
systems, where we considered the minimization of MSE subject to the
transmit power constraint and the unit modulus constraint on each
element of the RF analog BF matrices.  Due to the highly coupled
constraints, this optimization problem is hard to tackle.  We first
transformed it into a simpler form and then developed an innovative
BCD-based algorithm to solve it.  Besides, we proposed a novel
TOSCA-based two-timescale joint design algorithm to further reduce the
CSI signalling overhead and the effects of outdated CSI caused by the
severe delays.  These proposed algorithms can be guaranteed to obtain
the KKT solution of the original problem.  Moreover, with the aid of
the lower bound of the MSE, we also determined the near-optimal
cancellation order for the THP structure.  Our simulation results
demonstrated that the proposed THP-based hybrid transceiver design
algorithm can significantly outperform the existing linear hybrid
transceiver design algorithms and that the two-timescale joint design
algorithm has stronger robustness against the CSI delay than the
single-timescale algorithms.  Hence, the proposed BCD-based joint
design algorithm should be employed for the scenario of small CSI
delays, while the extended two-timescale joint design algorithm should
be applied for the case of severe CSI delays.

\begin{appendices}
\section{Proof of Theorem \ref{sigma_trans}}
\label{proof_sigma_trans}
It is readily seen that the KKT solution of problem
\eqref{original_problem} always makes the power constraint meet
equality.  By checking the first order conditions for problem
\eqref{trans_problem} and problem \eqref{original_problem}, it is
obvious that the scaled solution $\mathcal{\bar{S}}^{\star}$ is a KKT
solution of problem \eqref{trans_problem} and we can obtain the
following equations via comparing the derivatives of $\text{MSE}$ and
$\text{MSE}_{\sigma}$ at the point $\mathcal{\bar{S}}^{\star}$,
\begin{subequations}
	\begin{align*}
		&\frac{1}{\sigma_m}\frac{\partial\text{MSE}}{\partial \mathbf{P}_m^*}|_{\mathcal{\bar{S}}^{\star}}=\frac{\partial\text{MSE}_{\sigma}}{\partial \mathbf{P}_m^*}|_{\mathcal{\bar{S}}^{\star}}=\mathbf{0},\forall m,\\
		&\frac{\partial\text{MSE}}{\partial \mathbf{F}_m^*}|_{\mathcal{\bar{S}}^{\star}}+\boldsymbol{\lambda}_{F,m}\circ \mathbf{F}_{m}^{\star}=\frac{\partial\text{MSE}_{\sigma}}{\partial \mathbf{F}_m^*}|_{\mathcal{\bar{S}}^{\star}}+\boldsymbol{\lambda}_{F,m}\circ \mathbf{F}_{m}^{\star}=\mathbf{0},\forall m,\\
		&\frac{\partial\text{MSE}}{\partial \mathbf{T}^*}|_{\mathcal{\bar{S}}^{\star}}+\boldsymbol{\lambda}_{T}\circ \mathbf{T}^{\star}+\sum_{m=1}^M\frac{1}{P_t}\|\mathbf{\bar{P}}_m^{\star}\mathbf{F}_m^{\star}\|^2\mathbf{T}^{\star}\mathbf{\bar{W}}^{\star}(\mathbf{\bar{W}}^{\star})^H=\frac{\partial\text{MSE}_{\sigma}}{\partial \mathbf{T}^*}|_{\mathcal{\bar{S}}^{\star}}+\boldsymbol{\lambda}_{T}\circ \mathbf{T}^{\star}=\mathbf{0},\\
		&\frac{\partial\text{MSE}}{\partial \mathbf{W}^*}|_{\mathcal{\bar{S}}^{\star}}+\sum_{m=1}^M\frac{1}{P_t}\|\mathbf{\bar{P}}_m^{\star}\mathbf{F}_m^{\star}\|^2(\mathbf{T}^{\star})^H\mathbf{T}^{\star}\mathbf{\bar{W}}^{\star}=\frac{\partial\text{MSE}_{\sigma}}{\partial \mathbf{W}^*}|_{\mathcal{\bar{S}}^{\star}}=\mathbf{0},\\
		&\frac{\partial\text{MSE}}{\partial \mathbf{U}^*}|_{\mathcal{\bar{S}}^{\star}}=\frac{\partial\text{MSE}_{\sigma}}{\partial \mathbf{U}^*}|_{\mathcal{\bar{S}}^{\star}}=\mathbf{0},\\
		&|[\mathbf{F}_m]_{i,j}^{\star}|^2=1 \quad \forall m,i,j,\\	
		&|[\mathbf{T}]_{i,j}^{\star}|^2=1 \quad \forall i,j,\\
		&\|\mathbf{\bar{T}}\mathbf{\bar{W}}\|^2=P_t.
	\end{align*}
\end{subequations}
The above equations indicates that $\mathcal{S}^{\star}$ is a KKT solution of problem \eqref{original_problem} with the Lagrange multiplier attached to the power constraint being $\lambda_{P,T}=\sum_{m=1}^M\frac{1}{P_t}\|\mathbf{\bar{P}}_m^{\star}\mathbf{F}_m^{\star}\|^2$. This completes the proof.

\section{Derivation of gradients}
\label{derivation_of_gradients}
The partial derivatives with respect to the phase matrices can be
associated with the partial derivatives with respect to the analog BF
matrices by the following equations
\begin{equation}
	\frac{\partial g}{\partial \boldsymbol{\theta}_T}=\frac{\partial g}{\partial \mathbf{T}}\circ 1j\mathbf{T}-\frac{\partial g}{\partial \mathbf{T}^*}\circ 1j\mathbf{T}^*,\label{theta_T_partial}
\end{equation}
and
\begin{equation}
	\frac{\partial g}{\partial \boldsymbol{\theta}_{F_m}}=\frac{\partial g}{\partial \mathbf{F}_m}\circ 1j\mathbf{F}_m-\frac{\partial g}{\partial \mathbf{F}_m^*}\circ 1j\mathbf{F}_m^*.\label{theta_Fm_partial}
\end{equation}
Besides we have the expressions of the partial derivatives with respect to the analog BF matrices
\begin{equation}
\begin{aligned}
\frac{\partial g}{\partial \mathbf{T}^*} =& \sum_{m=1}^M\mathbf{\bar{H}}_m^H\mathbf{F}_m^H\mathbf{P}_m^H\mathbf{P}_m\mathbf{F}_m\mathbf{\bar{H}}_m\mathbf{TWW}^H+\sum_{m=1}^M\sigma_{e,m}^2\text{tr}(\mathbf{P}_m\mathbf{F}_m\mathbf{F}_m^H\mathbf{P}_m^H)\mathbf{TWW}^H\\
& -\sum_{m=1}^M\mathbf{\bar{H}}_m^H\mathbf{F}_m^H\mathbf{P}_m^H\mathbf{A}_m\mathbf{L}^T\mathbf{U}\mathbf{W}^H,
\end{aligned}\label{T_partial}
\end{equation}
and
\begin{equation}
\begin{aligned}
\frac{\partial g}{\partial \mathbf{F}_m^*} =& \mathbf{P}_m^H\mathbf{P}_m\mathbf{F}_m\mathbf{\bar{H}}_m\mathbf{T}\mathbf{WW}^H\mathbf{T}^H\mathbf{\bar{H}}_m^H+{\sigma}_{e,m}^2\text{tr}(\mathbf{TWW}^H\mathbf{T}^H)\mathbf{P}_m^H\mathbf{P}_m\mathbf{F}_m\\
&+\sigma_m^2\mathbf{P}_m^H\mathbf{P}_m\mathbf{F}_m-\mathbf{P}_m^H\mathbf{A}_m\mathbf{L}^T\mathbf{U}\mathbf{W}^H\mathbf{T}^H\mathbf{\bar{H}}_m^H.
\end{aligned}\label{Fm_partial}
\end{equation}
Finally, we can obtain the partial derivatives $\frac{\partial
  g}{\partial \boldsymbol{\theta}_T}$ and $\frac{\partial g}{\partial
  \boldsymbol{\theta}_{F_m}}$ by substituting \eqref{T_partial} and
\eqref{Fm_partial} into \eqref{theta_T_partial} and
\eqref{theta_Fm_partial}, respectively.

\end{appendices}

\normalem

\end{document}